\documentclass[11pt,letterpaper,abstract=true]{scrartcl}
\usepackage[T1]{fontenc}
\usepackage[utf8]{inputenc}
\usepackage{lmodern}
\usepackage[margin=1in]{geometry}
\usepackage{amsmath,amssymb,amsthm}
\usepackage{thmtools}
\usepackage{fontawesome5}
\usepackage{microtype}
\usepackage{booktabs,array,enumitem}
\usepackage{needspace}
\usepackage{mathtools}
\usepackage{dsfont}
\usepackage{xcolor}
\usepackage{tikz}
\usetikzlibrary{shapes.geometric}
\usepackage{todonotes}
\usepackage{bm}
\newcommand{\vect}[1]{\bm{#1}}
\newcommand{\enc}{\operatorname{enc}}
\newtheorem{theorem}{Theorem}[section]
\newtheorem{lemma}[theorem]{Lemma}
\newtheorem{proposition}[theorem]{Proposition}
\newtheorem{corollary}[theorem]{Corollary}
\newtheorem{claim}{Claim}[theorem]
\newenvironment{claimproof}{%
  \begin{proof}[Proof of claim]%
}{\end{proof}}

\theoremstyle{definition}
\newtheorem{definition}[theorem]{Definition}
\theoremstyle{remark}

\numberwithin{equation}{section}
\usepackage{hyperref}
\usepackage[nameinlink,noabbrev,capitalize]{cleveref}
\crefname{claim}{Claim}{claims}
\Crefname{claim}{Claim}{Claims}
\crefname{lemma}{Lemma}{Lemmas}
\Crefname{lemma}{Lemma}{Lemmas}

\usepackage[
  backend=biber,
  style=alphabetic,
  minalphanames=3,
  maxalphanames=4,
  maxbibnames=99,
  doi=true,
  giveninits=true
]{biblatex}
\hypersetup{colorlinks=true,linkcolor=blue!40!black,citecolor=blue!40!black,
  urlcolor=blue!40!black,
  pdftitle={A Fixed-Parameter Algorithm for 4-Block Integer Programming},
  pdfsubject={Finite penalties, periodic convexity, and fixed-parameter optimization}}

\newcommand{\Z}{\mathbb Z}
\newcommand{\R}{\mathbb R}

\newcommand{\conv}{\operatorname{conv}}
\newcommand{\aff}{\operatorname{aff}}

\newcommand{\dist}{\operatorname{dist}}

\newcommand{\abs}[1]{\lvert #1\rvert}
\newcommand{\norm}[1]{\lVert #1\rVert}
\newcommand{\ellloc}{{\vect\ell}^{(1:n)}}
\newcommand{\uloc}{{\vect u}^{(1:n)}}
\newcommand{\cloc}{{\vect c}^{(1:n)}}
\crefname{lemma}{Lemma}{Lemmas}
\Crefname{lemma}{Lemma}{Lemmas}

\title{\bfseries A Fixed-Parameter Algorithm for\\4-Block Integer Programming}
\author{Klaus Jansen\thanks{Kiel University, kj@informatik.uni-kiel.de, \url{https://orcid.org/0000-0001-8358-6796}}
\and
Felix Ohnesorge\thanks{Kiel University, foh@informatik.uni-kiel.de, \url{https://orcid.org/0009-0003-8023-3380}}
\and
Corinna Wambsganz\thanks{Kiel University, cwa@informatik.uni-kiel.de, \url{https://orcid.org/0009-0002-4820-9017}}
}
\date{}
\AfterEndEnvironment{abstract}{%
    \par\smallskip
    {\small\noindent\textbf{Keywords.} 4-Block, Integer Programming, Fixed-Parameter Tractable (FPT), separable convex\par
    \smallskip\noindent\textbf{Classification.} Theory of computation: fixed parameter tractability; integer programming\par
    \smallskip\noindent\textbf{Funding.} Deutsche Forschungsgemeinschaft (DFG, German Research Foundation) - Project number 528381760.}%
}

\begin{document}
\maketitle
\begin{abstract}
  We give a fixed-parameter tractable (FPT) algorithm with running time $f(k,\bar\Delta)\cdot {|I|}^{O(1)}$ for integer linear programs with 4-block structure, parameterized by the maximum block dimension $k$ and the largest absolute entry $\bar\Delta$ in the local-column blocks.
  This result resolves a long-standing open question in parameterized complexity, and answers a conjecture by Eisenbrand and Rothvoss (2026) in the positive.
  This result has several implications for other block-structured integer programming models, including mixed fracture number and special cases of 4-block programs with large entries outside of the diagonal.
  Important tools for this algorithm are structural properties of generalized $n$-fold integer programs shown by Ligthart~(2026) and an algorithm by Veselov, Gribanov, Zolotykh and Chirkov~(2020) for optimizing discrete convic functions. Independently and shortly after the first version of our paper appeared, Kouteck\'y, Lassota, and Ligthart (2026) obtained an FPT algorithm for 4-block integer programming.
\end{abstract}
\newpage
\section{Introduction}\label{sec:introduction}

Block-structured integer programs describe many small local decisions
linked by shared constraints. The 4-block model also introduces global
variables that appear in every local system, combining the coupling
by constraints of $n$-fold programming with the coupling by variables
of two-stage stochastic programming.

Let $\vect x\in\Z^k$ be the global variables and let $\vect y^{(i)}\in\Z^k$
be the variables of brick $i\in[n]=\{1,\ldots,n\}$. We consider
\begin{equation}\label{eq:fourblock}\tag{4-Block ILP}
  \begin{aligned}
    \min\quad & \langle \vect c^{(0)}, \vect x\rangle +\sum_{i=1}^n \langle \vect c^{(i)}, \vect y^{(i)}\rangle                       \\
    \text{s. t.}\quad
              & A^{(0)}\vect x+\sum_{i=1}^n B^{(i)} \vect y^{(i)}=\vect b^{(0)},                                                      \\
              & C^{(i)} \vect x+D^{(i)} \vect y^{(i)}=\vect b^{(i)}                                             &  & \forall i\in[n], \\
              & \vect \ell^{(0)}\le \vect x\le \vect u^{(0)},                                                                         \\
              & \vect \ell^{(i)}\le \vect y^{(i)}\le \vect u^{(i)}                                              &  & \forall i\in[n].
  \end{aligned}
\end{equation}
We use a common block dimension $k$: all blocks $A^{(0)},B^{(i)},C^{(i)},D^{(i)}$ belong to $\Z^{k\times k}$, and $\vect c^{(i)}, \vect b^{(i)}\in\Z^k$ for $i=0,\ldots,n$.
Smaller blocks are padded with zero rows and columns, with added variables fixed to zero and assigned zero cost.
Thus, the original blocks may have different dimensions and entries across bricks. We refer to $B^{(i)}$ and $D^{(i)}$ as local-column blocks.
Every constraint-matrix entry has absolute value at most $\Delta$, while $\bar\Delta$ denotes the maximum absolute entry in the local-column blocks $B^{(i)}$ and $D^{(i)}$.
All matrix entries, right-hand sides, and finite bounds are integral; the bounds may also be infinite.
We take $k\ge1$ and assume \(\Delta,\bar\Delta\geq 1\).
Costs are integral without loss of generality: rational costs can be multiplied by a common positive denominator.
Rational finite bounds can be rounded inward.
These preprocessing steps preserve polynomial encoding length.

For 4-blocks, Hemmecke, K\"oppe, and Weismantel established XP algorithms for linear~\cite{HKW10} and separable convex~\cite{HKW14} objectives with fixed uniform blocks, meaning that each family $B^{(i)},C^{(i)},D^{(i)}$ is constant across bricks.
Chen, Kouteck\'y, Xu, and Shi~\cite{CKXS20} improved augmentation bounds but also constructed instances whose nonzero integer kernel vectors have norms growing with $n$.
This obstructs augmentation algorithms that enumerate a parameter-sized box for the global part of an improving step, without establishing $\mathrm{W}[1]$-hardness.
Chen, Chen, and Zhang~\cite{CCZ24} obtained FPT, including for separable convex objectives, when all $C^{(i)}$ equal one matrix of rank at most one.
Oertel, Paat, and Weismantel~\cite{OPW24} derived proximity bounds for nonuniform 4-block programs using a colorful Steinitz lemma.
Lassota and Ligthart~\cite{LL26} subsequently improved the XP dependence on $n$ from $n^{k^2+O(1)}$ to $n^{k+O(1)}$ for square $k\times k$ blocks; their bounded-coefficient result also permits nonuniform blocks.

The gap between XP and FPT has been a recurring open question (see~\cite{Kou16,KLO18,CKXS20,Che19,EGK+19,EHK+25,GKK22,CKL+25,Kou25}).
In 2026, Eisenbrand and Rothvoss~\cite{ER26} formulated the explicit conjecture and say:
\begin{quote}
  \emph{``It is a popular open problem in the theoretical IP community whether there is an FPT-type algorithm.''}
\end{quote}
This work answers that important open question and confirms the conjecture.
\Cref{thm:main} establishes fixed-parameter tractability (FPT) for the nonuniform 4-block model parameterized by \(k\) and \(\bar\Delta\).
The algorithm returns an optimal integer solution, certifies infeasibility, or reports an objective unbounded below.
In \Cref{app:convex-comparison} we extend \Cref{thm:main} to separable convex objectives given by comparison oracles for finite variable bounds, using the same parameterization.
By the reduction of Dvo\v{r}\'ak, Eiben, Ganian, Knop and Ordyniak~\cite[Theorem~11]{DEG+21}, it also establishes FPT for ILP parameterized by the mixed fracture number of its incidence graph and the largest absolute matrix entry.
Moreover, using~\cite[Observation~1.3]{CKL+25}, our algorithm can be used to solve the feasibility problem on 4-block ILPs where the coefficients in the blocks $B^{(i)}=B$ can be larger if \(A^{(i)}=A\) and \(B^{(i)}=B\). 

\begin{restatable}[Main theorem]{theorem}{mainthm}\label{thm:main}
  Let \(I\) be an instance of~\eqref{eq:fourblock}.
  There is a computable function $f$ such that \(I\) can be solved in time
  \begin{equation*}
    f(k,\bar\Delta)\cdot {|I|}^{O(1)},
  \end{equation*}
  where \(|I|\) denotes the binary encoding length of \(I\).
\end{restatable}

Kouteck\'y, Lassota, and Ligthart~\cite{KLL26} independently obtained an FPT algorithm for 4-block integer programming shortly after the first version of this paper appeared.
Their results cover separable convex objectives given by comparison oracles and extend to matrices of bounded primal or dual treedepth with a bounded number of additional global columns.
They also give explicit parameter dependence and near-linear dependence on the number of local variables, measured in arithmetic operations and oracle queries.
Both approaches build on Ligthart's periodic-convexity framework.

Algorithmically, we optimize an everywhere-finite penalized value function over the global variables by exploiting periodic convexity and using $n$-fold optimization as a black box, thereby avoiding augmentation with Graver elements of the full 4-block matrix.
Applications of these block-structured models include scheduling~\cite{CCZ24}, transportation and production planning~\cite{DLOW08}, computational social choice and string algorithms for sequence consensus~\cite{KKM20}, statistics and disclosure control, and stochastic network planning~\cite{HKW14}.

To make the connection with $n$-fold programming explicit, stack the local variable vectors into $\vect y \coloneqq (\vect y^{(1)}, \ldots, \vect y^{(n)})^{\mathsf T}$.
The constraints separate into the global-variable part $Z$ and the $n$-fold part $E$:
\begin{equation*}
  Z\vect x+E\vect y=\vect b,\qquad
  Z=\begin{pmatrix}A^{(0)}\\C^{(1)}\\\vdots\\C^{(n)}\end{pmatrix},\qquad
  E=\begin{pmatrix}
    B^{(1)} & B^{(2)} & \cdots & B^{(n)} \\
    D^{(1)} & 0       & \cdots & 0       \\
    0       & D^{(2)} & \cdots & 0       \\
    \vdots  & \vdots  & \ddots & \vdots  \\
    0       & 0       & \cdots & D^{(n)}
  \end{pmatrix}.
\end{equation*}
The columns of $Z$ correspond to the global variables.
In $E$, the top blocks $B^{(i)}$ link the bricks and the diagonal blocks $D^{(i)}$ impose their local constraints.
We call $E$ a \emph{generalized $n$-fold matrix} because these blocks may differ across bricks.
Thus fixing $\vect x$ leaves the $n$-fold system $E\vect y=\vect b-Z\vect x$, with the original local bounds.
The challenge is to optimize over $\vect x$ while exploiting this tractable local problem.

For $n$-fold programming, De Loera, Hemmecke, Onn, and Weismantel~\cite{DLOW08} first obtained polynomial-time algorithms for fixed blocks, with a block-dependent exponent in $n$ (an XP bound).
Hemmecke, Onn, and Romanchuk~\cite{HOR13} obtained FPT with cubic dependence on $n$, and De Loera, Hemmecke, and Lee~\cite{DHL15} established strongly polynomial augmentation for fixed blocks.
Eisenbrand, Hunkenschr\"oder, and Klein~\cite{EHK18} improved parameter dependence and achieved nearly quadratic dependence on the number of bricks, including nonuniform blocks.
Kouteck\'y, Levin, and Onn~\cite{KLO18} obtained strongly polynomial FPT algorithms through primal and dual treedepth, developed further in the sparse-IP framework~\cite{EHK+25}.
Jansen, Lassota, and Rohwedder~\cite{JLR20} achieved near-linear dependence on the number of variables; Cslovjecsek, Eisenbrand, Hunkenschröder, Rohwedder and Weismantel~\cite{CEH+21} combined near-linear complexity with strong polynomiality for linear objectives.

The coefficient parameter is essential in general, already for restricted $n$-fold systems~\cite{CCZ22}, but selected blocks may admit large entries.
Cslovjecsek, Koutecký, Lassota, Pilipczuk and Polak~\cite{CKL+25} obtain FPT $n$-fold linear optimization with a common, possibly large linking block $B^{(i)}=B$ and varying bounded-entry diagonal blocks $D^{(i)}$;
unrelated large linking coefficients would encode Subset Sum.
They also obtain two-stage feasibility with arbitrary $C^{(i)}$ and bounded diagonal entries.
Their Observation~1.3 relates bounded-coefficient uniform 4-block systems to uniform systems whose non-diagonal entries are bounded by $n$.
Eisenbrand and Rothvoss~\cite{ER26} use residue-dependent affine descriptions of integer hulls to obtain two-stage optimization with large global coefficients;
their 4-block application allows additive violations of global equations.
Lassota and Ligthart's exact XP algorithm~\cite{LL26} also allows large $A^{(0)}$, common $B^{(i)}=B$, and varying $C^{(i)}$, parameterizing diagonal entries.

Our algorithm exploits periodic convexity of the optimum local cost as a function of the global variables. Ligthart~\cite{Lig26} develops this approach for separable convex objectives, including large-entry $n$-fold programs with a common linking block, and proves periodic convexity along lines for general $n$-fold value functions.
That paper also utilizes the algorithm by Veselov, Gribanov, Zolotykh and Chirkov~\cite{VGZC20}.
Extending this property to fixed-dimensional subspaces is the route to 4-block FPT identified there.
We supply this step for an \emph{everywhere-finite} penalized value function.
The structural input is the uniform integer-decomposition dilation for generalized $n$-fold matrices~\cite[Corollary~12]{Lig26}; here uniform refers to a common scale, not identical blocks.
\subsection*{Structure of the Paper}
The paper is structured as follows.
In~\Cref{sec:preliminaries} we introduce the necessary background and notation.
In~\Cref{sec:penalty} we define an everywhere-finite value function \(F(\vect x)\) that penalizes violations of the coupling constraints.
For finite variable bounds, minimizing this function over the global variables \(\vect x\) finds an optimal solution \((\vect x^\star, \vect y^\star)\) or certifies infeasibility.
Afterwards, in~\Cref{sec:convexity}, we show that \(F(\vect x)\) is convex extensible on each residue class modulo \(M = kN\), where \(N\) depends only on \(k\) and \(\bar\Delta\).
These structural properties allow us to optimize \(F(\vect x)\) over each residue class separately, which is the content of~\Cref{sec:algorithm}.
Afterwards, we show how to reduce the general case with unbounded variables to the case with finite variable bounds.

\section{Preliminaries}\label{sec:preliminaries}

For vectors $\vect a, \vect b \in \R^d$, vector inequalities $\vect a \le \vect b$, absolute values $\abs{\vect a}$, and floor/ceiling roundings $\lfloor \vect a\rfloor, \lceil \vect a\rceil$ are understood componentwise.
The vector of all ones is denoted $\mathds{1}$, and $I_d$ denotes the $d \times d$ identity matrix (we write $I$ when the dimension is clear from context, distinguishing it from the instance $I$).
The Euclidean ball of radius $\rho \ge 0$ centered at the origin is $B_\rho \coloneqq \{\vect z \in \R^d : \norm{\vect z}_2 \le \rho\}$.

For a matrix $M \in \R^{p \times q}$, $\abs{M}$ denotes the entrywise absolute value matrix $(\abs{M_{ij}})_{i,j}$, and $\norm{M}_{\max} \coloneqq \max_{i,j} \abs{M_{ij}}$ is the maximum absolute entry.
The stacked vector of local variables is $\vect y \coloneqq ((\vect y^{(1)})^{\mathsf T}, \ldots, (\vect y^{(n)})^{\mathsf T})^{\mathsf T} \in \Z^{nk}$ and we denote the corresponding upper and lower bound vectors by \(\ellloc \leq \vect y \leq \uloc\).
The local cost vector is defined analogously by \(\cloc \coloneqq ((\vect c^{(1)})^{\mathsf T}, \ldots, (\vect c^{(n)})^{\mathsf T})^{\mathsf T}\).

For a set $S \subseteq \R^d$, we write $\conv(S)$, $\aff(S)$, and $\operatorname{cone}(S)$ for the convex hull, affine hull, and conical hull (the set of all finite nonnegative linear combinations of elements of $S$), respectively.
Vectors $\vect p_0, \ldots, \vect p_r \in \R^d$ are \emph{affinely independent} if $\sum_{i=0}^r \lambda_i \vect p_i = \vect 0$ with $\sum_{i=0}^r \lambda_i = 0$ implies $\lambda_0 = \cdots = \lambda_r = 0$.
For a nonempty polyhedron $P = \{\vect z \in \R_{\ge0}^d : H\vect z = \vect b\}$, its \emph{recession cone} is $\operatorname{rec}(P) \coloneqq \{\vect g \in \R_{\ge0}^d : H\vect g = \vect 0\}$.

\begin{definition}[Conformal order]\label{def:conformal}
  For vectors $\vect a,\vect b\in\R^d$, write $\vect a\sqsubseteq \vect b$ if
  $a_j b_j\ge0$ and $|a_j|\le|b_j|$ for every coordinate $j\in[d]$.
  We say that $\vect a$ is \emph{conformal to} $\vect b$.
\end{definition}

\begin{definition}[Graver elements and Graver basis]\label{def:graver}
  For an integer matrix $A \in \Z^{m \times d}$, the \emph{Graver basis} $\mathcal{G}(A)$ is the set of all $\sqsubseteq$-minimal non-zero vectors in $\{\vect z \in \Z^d : A\vect z = \vect 0\}$. Its elements are called \emph{Graver elements}.
\end{definition}

\begin{definition}[Midpoint convexity and convex extensibility]
  \label{def:convexity}
  A function $h:\Z^k\to\R$ is \emph{midpoint-convex} if
  \begin{equation*}
    h\left(\frac{\vect a+\vect b}{2}\right)\le\frac{h(\vect a)+h(\vect b)}2
    \quad\text{whenever }\vect a,\vect b\in\Z^k\text{ and }\vect a+\vect b\in2\Z^k.
  \end{equation*}
  It is \emph{convex extensible} if there is an everywhere-finite convex
  function $\widehat h:\R^k\to\R$ such that $\widehat h(\vect z)=h(\vect z)$ for all
  $\vect z\in\Z^k$.
\end{definition}

\begin{definition}[Discrete convic function~\cite{JOGO,VGZC20}]\label{def:discreteconvic}
  Let $D \subseteq \Z^k$. A function $f:D \to \R$ is called \emph{discrete convic} if for any $\vect y \in D$ and any $\vect z \in D$ belonging to the set
  \[
    \vect y + \operatorname{cone}\{\vect y - \vect x : \vect x\in D,\ f(\vect x) \le f(\vect y)\},
  \]
  it holds that $f(\vect y) \le f(\vect z)$.
\end{definition}

\begin{lemma}\label{lem:gmidpoint}
  Let $h:\Z^k \to \R$ be a midpoint-convex function and let $A:\R^k \to \R$ be an affine function. Then the function
  \[
    g:\Z^k \to \R,\qquad \vect v \mapsto h(\vect v) - A(\vect v)
  \]
  is midpoint-convex.
\end{lemma}
\begin{proof}
  Let $\vect u, \vect v \in \Z^k$ such that $\vect u + \vect v \in 2\Z^k$, and let $\vect w = (\vect u + \vect v)/2 \in \Z^k$ be their integer midpoint.
  By the midpoint convexity of $h$ and the affinity of $A$,
  \[
    g(\vect w)
    = h(\vect w) - A(\vect w)
    \le \frac{h(\vect u) + h(\vect v)}{2} - \frac{A(\vect u) + A(\vect v)}{2}
    = \frac{g(\vect u) + g(\vect v)}{2}.
  \]
  Thus, $g$ is midpoint-convex.
\end{proof}

\subsection*{Algorithmic and Structural Tools}

A key ingredient for this algorithm is the conformal bisection (\cref{cor:bisection}) for generalized $n$-fold matrices of Ligthart~\cite[Observation~5]{Lig26}. We restate it here in a form that is convenient for our purposes by first restating an underlying property also stated in their work.

A polyhedron $P$ has the \emph{integer decomposition property} (IDP) if
every integer point of $aP$, for every positive integer $a \in \mathbb{Z}_{> 0}$, is a sum
of $a$ integer points of $P$.

\begin{proposition}[Uniform $n$-fold IDP dilation, {\cite[Corollary~12]{Lig26}}]
  \label{prop:idp}
  For every $k,\bar\Delta \geq 1$, there exists a computable positive integer \(N = 2^{{(k\bar\Delta)}^{O(k^2)}}\) such that for every generalized $n$-fold matrix $Q$ with blocks of size at most $k \times k$ and entries of absolute value at most $\bar\Delta$, and for every integral vector $\vect w$, the polyhedron
  \[
    \{\vect z \ge \vect 0 : Q\vect z = N\vect w\}
  \]
  has the IDP.
\end{proposition}

\begin{corollary}[Conformal bisection, {\cite[Observation 5]{Lig26}}]\label{cor:bisection}
  Let $Q$ be a generalized $n$-fold matrix and let \(N \in \mathbb{Z}_{> 0}\) be the positive integer of \Cref{prop:idp}. For all integral vectors $\vect \delta,\vect w$ with $Q\vect \delta=2N\vect w$, there exist integral vectors \(\vect \delta', \vect \delta''\) such that
  \[
    \vect \delta=\vect \delta'+\vect \delta'',\qquad
    \vect\delta', \vect \delta''\sqsubseteq\vect \delta,\qquad
    Q\vect\delta'=Q\vect\delta''=N\vect w.
  \]
\end{corollary}
\begin{proof}
  Change the sign of each column for which the corresponding coordinate of $\vect \delta$ is negative.
  The resulting matrix $\widetilde Q$ remains in the same generalized $n$-fold class and satisfies $\widetilde Q|\vect \delta|=2N\vect w$.
  Apply the IDP to $|\vect \delta|\in2P$ for $P=\{\vect z\ge0:\widetilde Q\vect z=N\vect w\}$, and undo the sign changes.
  The two nonnegative summands of $|\vect \delta|$ are coordinatewise bounded by $|\vect \delta|$, giving the required conformality.
  This is the consequence of IDP used in \cite[Observation 5]{Lig26}.
\end{proof}

In this paper, we utilize two black-box algorithms.
First, an FPT algorithm for generalized $n$-folds~\cite{EHK18,CEH+21,KLO18}, and second an algorithm by Veselov, Gribanov, Zolotykh and Chirkov~\cite{VGZC20} for minimizing discrete convic functions.
\begin{theorem}[Generalized $n$-fold optimization]
  \label{thm:nfold}
  There is a computable function $f$ such that a linear generalized $n$-fold integer program $I$ with integral data, finite variable bounds, $k\times k$ blocks, and matrix entries of absolute value at most $\bar\Delta$, where $k,\bar\Delta\ge1$, can be solved in time
  \begin{equation*}
    f(k,\bar\Delta)\cdot |I|^{O(1)}.
  \end{equation*}
  The algorithm returns an optimal integer solution and its value, or reports infeasibility.
\end{theorem}
Writing $B_\rho=\{\vect z\in\R^d:\norm{\vect z}_2\le\rho\}$ for the Euclidean ball of radius $\rho$ centered at the origin, the following theorem allows us to minimize a discrete convic function over its integer points using only comparisons of function values, with complexity polynomial in $\log\rho$ for fixed dimension $d$.
\begin{theorem}[{\cite[Theorem 1]{VGZC20}}]\label{thm:veselov}
  There exists an algorithm for minimizing a discrete convic function, given by the comparison oracle on $B_{\rho} \cap \Z^d$, using $2^{O(d^2 \log (d+1))} \log (2+\rho)$ calls to the oracle and $2^{O(d^2 \log (d+1))} (\log (2+\rho))^{O(1)}$ arithmetic operations.
\end{theorem}
\section{The Penalized Value Function}\label{sec:penalty}

For now assume that all lower and upper bounds \(\vect \ell, \vect u\) are finite.
W.l.o.g.\ assume that no interval is empty (i.e., \(\ell_j^{(i)}\le u_j^{(i)}\)).
Now define:
\begin{equation}\label{eq:penalty-weight}
  W=\sum_{i=0}^{n}\sum_{j=1}^k |c_j^{(i)}|(u_j^{(i)}-\ell_j^{(i)}),\qquad K=W+1.
\end{equation}
Here, $W$ is the maximum variation of the objective function.
For a global vector $\vect x$ and a local vector $\vect y$, define
their residual vector by
\[
  \vect s=\vect b-Z\vect x-E\vect y.
\]
Thus $\vect s=\vect 0$ exactly when the coupling equations are satisfied.
We use the residual-penalty construction from the proof of~\cite[Lemma 18]{Lig26}.
For the bounded linear objective considered here, we choose the penalty coefficient explicitly as $K = W+1$.
This permits numerical evaluation of the penalized value function, defined for every $\vect x\in\Z^k$ (including infeasible points):
\begin{equation}\label{eq:value}
  F(\vect x)
  =\langle\vect c^{(0)},\vect x\rangle
  +\min\left\{
  \sum_{i=1}^n
  \langle\vect c^{(i)},\vect y^{(i)}\rangle
  +K\norm{\vect s}_1:
  \begin{array}{l}
    E\vect y+\vect s=\vect b-Z\vect x, \\
    \ellloc \leq \vect y \leq \uloc,\
    \vect s\in\Z^{k(n+1)}
  \end{array}
  \right\}.
\end{equation}
After grouping the residual variables into bricks, the matrix $Q=[E\ I]$ has generalized $n$-fold structure with one additional brick.
Residual variables for local equations join their corresponding bricks; those for global equations form the additional brick, whose local block is zero.
Every brick has width at most $2k$, and every matrix entry has absolute value at most $\bar\Delta$.
Padding with zero rows makes the local row dimensions uniform.

Given \(\vect x\), for each admissible $\vect y$, the equality uniquely determines $\vect s$.
Therefore, this minimum ranges over finitely many pairs $(\vect y,\vect s)$, despite the absence of explicit bounds on $\vect s$.
The residual penalty lets us evaluate $F$ even when fixing $\vect x$ makes the original equations infeasible.
This is necessary because the finite-function lemma in \cref{sec:finite-lemma} does not allow $+\infty$ values.

\begin{lemma}[Exact penalty]\label{lem:exact}
  For every \(\vect x\in\Z^k\), the value \(F(\vect x)\) is finite, and the minimum in \eqref{eq:value} is attained by a pair \((\vect y,\vect s)\).
  Let \(\vect x^\star\) minimize \(F\) with \(\vect \ell^{(0)} \leq \vect x^\star \leq \vect u^{(0)}\), and let \((\vect y^\star,\vect s^\star)\) attain the corresponding inner minimum.
  Then \eqref{eq:fourblock} is feasible if and only if
  \[
    Z\vect x^\star+E\vect y^\star=\vect b, \qquad \ellloc \leq \vect y^\star \leq \uloc.
  \]
  In this case, \((\vect x^\star,\vect y^\star)\) is an
  optimal solution of \eqref{eq:fourblock}.
\end{lemma}
\begin{proof}
  Assume w.l.o.g.\ that \(\ellloc \leq \uloc\).
  For fixed $\vect x$, each local vector \(\vect y\) determines exactly one residual vector \(\vect s\), so the minimum defining $F(\vect x)$ exists and is finite.

  Suppose a feasible vector $(\vect x,\vect y)$ to~\eqref{eq:fourblock} exists. Let \(\vect c=((\vect c^{(0)})^T, (\vect c^{(1:n)})^T)^T\).
  Every infeasible pair $(\vect{\bar x},\vect{\bar y}) \in [\vect \ell^{(0)},\vect u^{(0)}] \times [\vect \ell^{(1:n)},\vect u^{(1:n)}]$ to~\eqref{eq:fourblock} has a nonzero integral residual \(\vect {\bar s}\) and hence
  \[
    \langle \vect c, (\vect{\bar x}, \vect{\bar y})\rangle+K\norm{\vect b-Z\vect{\bar x}-E\vect{\bar y}}_1
    \ge \langle \vect c, (\vect x,\vect y)\rangle-W+K
    = \langle \vect c, (\vect x,\vect y) \rangle+1.
  \]
  Thus every minimizer of the penalized objective in~\eqref{eq:value} is a feasible solution to~\eqref{eq:fourblock} whenever~\eqref{eq:fourblock} is feasible.

  On feasible pairs \((\vect x, \vect y)\) to~\eqref{eq:fourblock} the term \(K\norm{\vect b-Z\vect{ x}-E\vect{ y}}_1\) equals zero, so minimizing the penalized objective in~\eqref{eq:value} minimizes the original objective.
  Finally, minimizing~\eqref{eq:value} subject to \(\vect \ell^{(0)} \leq \vect x \leq \vect u^{(0)}\) and recovering an attaining local vector \(\vect y\) is equivalent to jointly minimizing this penalized objective over both boxes.
  Both conclusions follow.
\end{proof}
\begin{lemma}[Value and witness oracle]\label{lem:oracle}
  Given $\vect x\in\Z^k$, one can compute $F(\vect x)$ and an attaining local vector $\vect y$ in time
  \[
    f(k,\bar\Delta) \cdot {(|I| + \enc(\vect x))}^{O(1)},
  \]
  where \(f\) is a computable function.
\end{lemma}
\begin{proof}
  To obtain the desired running time, we want to apply~\cref{thm:nfold}.
  For this, we have to show that the function \(F(\vect x)\) in~\eqref{eq:value} can be formulated as a linear \(n\)-fold.
  This can be done with standard techniques.
  We provide the full proof for completeness.

  Fix $\vect x\in\Z^k$.
  We first show how to linearize the objective.
  Set
  \begin{equation*}
    \vect s^{(0)}(\vect x)=\vect b-Z\vect x-E\ellloc,\qquad
    \vect S(\vect x)=|\vect s^{(0)}(\vect x)|+|E|(\uloc-\ellloc),
  \end{equation*}
  where \(\vect S(\vect x)\) is the vector of slack bounds.
  If $\ellloc\le\vect y\le\uloc$, then $0\le\vect y-\ellloc\le\uloc-\ellloc$.
  Applying the triangle inequality componentwise gives
  \[
    \begin{aligned}
      |\vect b-Z\vect x-E\vect y|
       & =|\vect s^{(0)}(\vect x)-E(\vect y-\ellloc)|        \\
       & \le |\vect s^{(0)}(\vect x)|+|E|\,|\vect y-\ellloc| \\
       & =|\vect s^{(0)}(\vect x)|+|E|(\vect y-\ellloc)      \\
       & \le |\vect s^{(0)}(\vect x)|+|E|(\uloc-\ellloc)
      =\vect S(\vect x).
    \end{aligned}
  \]
  The last inequality uses the nonnegativity of every entry of $|E|$.
  Consider the bounded linear integer program
  \begin{equation}\label{eq:linear-oracle}
    \begin{aligned}
      \min\quad         & \langle \cloc, \vect y \rangle
      +\langle (K\mathds{1}), (\vect s^++\vect s^-)\rangle               \\
      \text{s. t.}\quad & E\vect y+\vect s^+-\vect s^-=\vect b-Z\vect x, \\
                        & \ellloc\le\vect y\le\uloc,\qquad
      0\le\vect s^+,\vect s^-\le\vect S(\vect x),                        \\
                        & \vect y,\vect s^+,\vect s^-\text{ integral}.
    \end{aligned}
  \end{equation}
  For each $\vect y$, the positive and negative parts of $\vect b-Z\vect x-E\vect y$ satisfy the slack bounds and minimize $\langle (K\mathds{1}), (\vect s^++\vect s^-)\rangle$ to $\norm{\vect b-Z\vect x-E\vect y}_1$.
  Indeed, $s_j^++s_j^-\ge|s_j^+-s_j^-|$ for each coordinate.
  Thus \eqref{eq:linear-oracle} is feasible and its optimum equals the inner minimum in \eqref{eq:value}.
  Adding $\langle\vect c^{(0)},\vect x\rangle$ gives $F(\vect x)$, and the optimal $\vect y$ is the required witness.

  To see the generalized $n$-fold structure, group the slack variables for each local system with its original brick.
  The resulting global and local blocks of brick $i$ are
  \[
    [B^{(i)}\ 0\ 0]\qquad\text{and}\qquad[D^{(i)}\ I_k\ {-I_k}].
  \]
  The global-row slacks form one additional brick with global block $[I_k\ {-I_k}]$ and local equations $0=0$. Hence there are $n+1$ bricks, with $k$ global rows, $k$ local rows per brick after padding, and at most $3k$ variables per brick. All matrix entries have absolute value at most $\bar\Delta$.

  Applying~\cref{thm:nfold} proves the claimed running time.
\end{proof}

\section{Periodic Convexity}
\label{sec:convexity}
We now establish the periodic convexity of \(F\), which will allow us to optimize over the global variables one residue class at a time.
We first prove midpoint convexity on residue classes modulo \(N\) and then obtain convex extensibility on residue classes modulo \(M=kN\).

\subsection{Midpoint Convexity via Conformal Bisection}

To prove midpoint convexity of $F$, we apply conformal bisection simultaneously to the local variables and their residuals. Note that $Q=[E\ I]$ is a generalized $n$-fold matrix. Let $N$ be the positive integer from \Cref{prop:idp} for $Q$. We define for each $\vect r\in\{0,\ldots,N-1\}^k$ the following function:
\begin{align}\label{eq:def-h}
  H_{\vect r}: \mathbb{Z}^k \rightarrow\mathbb{Z}, \vect v \mapsto F(\vect r + N\vect v).
\end{align}

\begin{lemma}[Periodic midpoint convexity]\label{lem:periodic-midpoint}
  For each $\vect r\in\{0,\ldots,N-1\}^k$, the function \(H_{\vect r}\) defined in~\eqref{eq:def-h} is finite and midpoint-convex.
\end{lemma}
\begin{proof}
  Finiteness follows from \Cref{lem:exact}.
  Let $\vect a',\vect a'' \in\Z^k$ have integer midpoint $\vect w=(\vect a'+\vect a'')/2$.
  Our goal is now to show midpoint-convexity, i.e.,
  \[
    2F(\vect r + N\vect w) \leq F(\vect r + N \vect a') + F(\vect r + N \vect a'').
  \]
  Choose optimal witnesses $\vect z'=(\vect y',\vect s')$ and $\vect z''=(\vect y'',\vect s'')$ in \eqref{eq:value} for the global vectors $\vect r+N\vect a'$ and $\vect r+N\vect a''$.
  Their difference $\vect \delta=\vect z''- \vect z'$ satisfies
  \[
    Q\vect \delta=-NZ(\vect a''-\vect a')
    =2N\left(-Z\frac{\vect a''-\vect a'}{2}\right).
  \]
  Note that the vector \(-Z\frac{\vect a''-\vect a'}{2}\) is integral. By \cref{cor:bisection}, write
  $\vect \delta=\vect \delta'+ \vect \delta''$ with
  $\vect \delta', \vect \delta''\sqsubseteq\vect \delta$ and
  $Q\vect \delta'=Q \vect \delta'' =-NZ(\vect a''-\vect a')/2$. Define
  \[
    \vect{\widehat z}'=\vect z'+\vect \delta',\qquad
    \vect{\widehat z}''=\vect z'+\vect \delta''.
  \]
  Write the new points as \(\widehat{\vect z}'=(\widehat{\vect y}',\widehat{\vect s}')\) and \(\widehat{\vect z}''=(\widehat{\vect y}'',\widehat{\vect s}'')\).
  Both points satisfy the equations at the midpoint: For $\vect{\widehat z}'$, we have
  \begin{align*}
    Q\widehat{\vect z}' & = Q(\vect z' + \vect \delta') = Q\vect z' + Q\vect \delta'               \\
                        & = \vect b - Z(\vect r + N\vect a') + N(-Z\frac{\vect a'' - \vect a'}{2}) \\
                        & = \vect b - Z\vect r - ZN\vect a' - ZN\frac{\vect a'' - \vect a'}{2}     \\
                        & = \vect b - Z\vect r - ZN(\vect a' + \frac{\vect a'' - \vect a'}{2})     \\
                        & = \vect b - Z\vect r - ZN (\frac{\vect a' + \vect a''}{2})               \\
                        & = \vect b-Z(\vect r + N(\frac{\vect a' + \vect a''}{2}))                 \\
                        & =\vect b-Z(\vect r+N\vect w).
  \end{align*}
  Since $Q\widehat{\vect z}'' = Q(\vect z' + \vect \delta'') = Q\vect z' + Q\vect \delta'' =Q\vect z' + Q\vect \delta'$, we get analogously $Q\widehat{\vect z}''=\vect b-Z(\vect r+N\vect w)$.
  Conformality places each coordinate of either new point between the corresponding coordinates of \(\vect z'\) and \(\vect z''\).
  In particular, $\ellloc \leq \vect {\widehat y}', \vect {\widehat y}'' \leq \uloc$, so the new points are feasible for the minimum defining \(F(\vect r+N\vect w)\).
  Moreover,
  \begin{equation}\label{eq:sum-witnesses}
    \widehat{\vect z}'+\widehat{\vect z}''
    =2\vect z'+\vect\delta'+\vect\delta''
    =2\vect z'+\vect\delta
    =\vect z'+\vect z''.
  \end{equation}

  We next show that the total cost of the two new points is at most that of the original witnesses.
  For a point
  \(\vect z=(\vect y,\vect s)\), write
  \[
    \Psi(\vect z)
    =\sum_{i=1}^n
    \langle\vect c^{(i)},\vect y^{(i)}\rangle
    +K\norm{\vect s}_1.
  \]
  By \eqref{eq:sum-witnesses}, the local components satisfy
  \[
    \widehat{\vect y}'+\widehat{\vect y}''
    =\vect y'+\vect y''.
  \]
  Their total linear local cost is therefore unchanged.

  For the residual penalty \(\hat {\vect s}'\) and \(\hat {\vect s}''\), fix a coordinate \(h\).
  The values \(\widehat s'_h,\widehat s''_h\) lie between
  \(s'_h,s''_h\) and have the same sum. Consequently, for
  some \(\lambda_h\in[0,1]\),
  \[
    \widehat s'_h
    =\lambda_h s'_h+(1-\lambda_h)s''_h,
    \qquad
    \widehat s''_h
    =(1-\lambda_h)s'_h+\lambda_h s''_h.
  \]
  By the convexity of the absolute value function, we have
  \begin{equation*}
    |\widehat s'_h|
    \leq \lambda_h|s'_h|+(1-\lambda_h)|s''_h|, \qquad \text{and }\qquad
    |\widehat s''_h|
    \leq (1-\lambda_h)|s'_h|+\lambda_h|s''_h|.
  \end{equation*}
  Addition gives:
  \[
    |\widehat s'_h|+|\widehat s''_h|
    \le |s'_h|+|s''_h|.
  \]
  Summing over \(h\), and using the equality of the total
  linear costs, yields
  \begin{equation}\label{eq:residual-penalty}
    \Psi(\widehat{\vect z}')+\Psi(\widehat{\vect z}'')
    \le \Psi(\vect z')+\Psi(\vect z'').
  \end{equation}

  We can now conclude midpoint convexity.
  Each new point is feasible for the minimum in~\eqref{eq:value} defining \(F(\vect r+N\vect w)\), therefore for both \(\widehat{\vect z} \in \{\widehat{\vect z}', \widehat{\vect z}''\}\):
  \begin{align}\label{eq:midpoint-witness-bound}
    F(\vect r+N\vect w)
    \le
    \langle\vect c^{(0)},\vect r+N\vect w\rangle
    +\Psi(\widehat{\vect z}).
  \end{align}
  Adding these two bounds and applying \eqref{eq:residual-penalty}, we obtain
  \[
    \begin{aligned}
      2F(\vect r+N\vect w)
                                                          & \le
      2\langle\vect c^{(0)},\vect r+N\vect w\rangle
      +\Psi(\widehat{\vect z}')+\Psi(\widehat{\vect z}'') & \text{by~\eqref{eq:midpoint-witness-bound}}  \\
                                                          & \le
      2\langle\vect c^{(0)},\vect r+N\vect w\rangle
      +\Psi(\vect z')+\Psi(\vect z'')                     & \text{by~\eqref{eq:residual-penalty}}        \\
                                                          & =\bigl(
      \langle\vect c^{(0)},\vect r+N\vect a'\rangle
      +\Psi(\vect z')\bigr) +\bigl(
      \langle\vect c^{(0)},\vect r+N\vect a''\rangle
      +\Psi(\vect z'')\bigr)                                                                             \\
                                                          & =F(\vect r+N\vect a')+F(\vect r+N\vect a'').
    \end{aligned}
  \]
  The first equality uses \(\vect w=(\vect a'+\vect a'')/2\) and linearity of the global cost.
  The last equality uses optimality of the original witnesses.
  This proves the required midpoint inequality (see~\cref{def:convexity}).
\end{proof}

\subsection{The Convex Extension}\label{sec:finite-lemma}

Midpoint convexity alone gives inequalities only for pairs with an
integer midpoint. The next lemma converts those inequalities into all
convex-combination inequalities after passing to a coarser lattice.

\begin{lemma}[Finite midpoint-convex functions on a coarser lattice]
  \label{lem:finite-function}
  Let $k\ge1$ and let $h:\Z^k\to\R$ be everywhere finite and
  midpoint-convex. For every $\vect \tau\in\Z^k$, the function
  \[
    \vect v\longmapsto h(\vect \tau+k\vect v)
  \]
  is convex extensible.
\end{lemma}
\begin{proof}
  Translation preserves the hypotheses, so it suffices to treat $\vect \tau=\mathbf{0}$.
  We establish four claims: a geometric property of lattice simplices, the Jensen inequality on each simplex, its extension to arbitrary finite convex combinations, and the existence of a finite convex envelope.

  \begin{restatable}[\faCut, {\cite[Theorem 3.1]{POWERS2021106608}}]{claim}{simplexMidpoints}\label{clm:simplex-midpoints}
    Let
    \[
      P=\conv\{k\vect p_0,\ldots,k\vect p_r\},\qquad
      \vect p_0,\ldots,\vect p_r\in\Z^k,\quad 0\le r\le k,
    \]
    where the vectors $\vect p_0,\ldots,\vect p_r$ are affinely independent.
    Every nonvertex $\vect w\in P\cap\Z^k$ is the midpoint of two distinct lattice points of $P$.
  \end{restatable}
  \begin{claimproof}
    The claim follows from Powers and Reznick~\cite[Theorem~3.1]{POWERS2021106608} by applying their theorem to the doubled simplex with dilation factor $k$ and dividing the resulting even lattice points by two; the case $k=1$ is immediate.
    For completeness, we provide a self-contained proof in~\Cref{sec:omitted-proofs}.
  \end{claimproof}

  \begin{claim}[Jensen inequality on a simplex]\label{clm:simplex-jensen}
    Let $P$ be a simplex as in \cref{clm:simplex-midpoints}, and let
    $A$ be the affine function on $\aff(P)$ agreeing with $h$ at the
    vertices of $P$. Then
    \begin{equation*}
      h(\vect w)\le A(\vect w)\qquad(\vect w\in P\cap\Z^k).
    \end{equation*}
  \end{claim}
  \begin{claimproof}
    We prove this by contradiction. Suppose $h-A$ has a positive maximum on the finite set $P\cap\Z^k$.
    Choose a maximizer $\vect w$ with largest squared Euclidean norm. Since $A$ agrees with $h$ at every vertex of $P$ by definition, we get $h(\vect p)-A(\vect p) = 0$ for every vertex $\vect p \in P$. Since $\vect w$ is a maximizer of $h-A$ and $h-A$ has a positive maximum, we get that $\vect w$ is not a vertex.
    By \cref{clm:simplex-midpoints}, $\vect w=(\vect u+\vect v)/2$ for distinct $\vect u,\vect v\in P\cap\Z^k$.

    We now show that $\vect u$ and $\vect v$ attain the same maximum. Set $g = h-A$. By \Cref{lem:gmidpoint}, we get that $g$ is midpoint-convex. By midpoint convexity (\cref{def:convexity}), we get $g(\vect w) \leq \frac{g(\vect u) + g(\vect v)}{2}$. Since $\vect w$ maximizes $g$, we get $g(\vect u), g(\vect v) \leq g(\vect w)$. Together, we have
    \[g(\vect w) \leq \frac{g(\vect u) + g(\vect v)}{2} \leq g(\vect w)\]
    and thus, $g(\vect u) = g(\vect v) = g(\vect w)$.
    But
    \begin{align*}
      \frac{\norm{\vect u}_2^2+\norm{\vect v}_2^2}{2}
      =\frac12\sum_{j=1}^k(u_j^2+v_j^2) & =\frac14\sum_{j=1}^k(2u_j^2+2v_j^2)                                           \\
                                        & =\frac14\sum_{j=1}^k(u_j^2+2u_jv_j+v_j^2)
      +\frac14\sum_{j=1}^k(u_j^2-2u_jv_j+v_j^2)                                                                         \\
                                        & =\frac14\sum_{j=1}^k(u_j+v_j)^2
      +\frac14\sum_{j=1}^k(u_j-v_j)^2                                                                                   \\
                                        & =\sum_{j=1}^k{\frac{(u_j+v_j)^2}{2^2}}
      +\frac14\sum_{j=1}^k(u_j-v_j)^2                                                                                   \\
                                        & =\norm{\vect w}_2^2+\frac{\norm{\vect u-\vect v}_2^2}{4} >\norm{\vect w}_2^2,
    \end{align*}

    contradicting the choice of $\vect w$. Thus the claim holds.
  \end{claimproof}

  \begin{claim}[Jensen inequality on the coarser lattice]\label{clm:lattice-jensen}
    For any $\vect p_0,\ldots,\vect p_a\in k\Z^k$ and nonnegative weights $\lambda_0,\ldots,\lambda_a \in \R_{\geq 0}$ with $\sum_i\lambda_i=1$, if $\vect z=\sum_i\lambda_i \vect p_i\in k\Z^k$, then
    \[
      h(\vect z)\le\sum_{i=0}^a\lambda_i h(\vect p_i).
    \]
  \end{claim}
  \begin{claimproof}
    Fix points \(\vect p_0, \ldots, \vect p_a \in k\Z^k\) and weights \(\lambda_0, \ldots, \lambda_a \in \R_{\geq 0}\) with \(\sum \lambda_i = 1\). Let $\vect z = \sum_i \lambda_i\vect p_i \in k\Z^k$.

    Consider the following linear program:
    \begin{equation*}
      \min\left\{\sum_{i=0}^{a} \bar{\lambda}_i h(\vect p_i) : \quad\sum_{i=0}^a \bar{\lambda}_i \vect p_i = \vect z, \quad\sum \bar \lambda_i = 1, \quad\bar{\lambda}_i \geq 0 \right\}.
    \end{equation*}
    Since $\vect z = \sum_i \lambda_i\vect p_i \in k\Z^k$, the solution polytope of this linear program is nonempty.
    Since an optimal solution can always be found at a vertex of the solution polytope, there exists an optimal solution with at most \(k+1\) nonzero variables with affinely independent support.
    Assume w.l.o.g.\ that \(\bar{\lambda}^\star_0, \ldots, \bar{\lambda}^\star_{r}\) are the nonzero variables of this solution with \(r \leq k\).
    Then $\vect z = \sum_{i=0}^{r} \bar{\lambda}^\star_i\vect p_i \in k\Z^k$ is an affinely independent representation of \(\vect z\).
    Now, define \(P = \conv\{\vect p_0, \ldots, \vect p_r\}\).
    Let $A$ be the unique affine function on $\aff(P)$ satisfying $A(\vect p_i)=h(\vect p_i)$ for $i=0,\ldots,r$.
    Such a function exists because the vertices are affinely independent.
    Since \(\vect p_i \in k\Z^k\), we can apply~\cref{clm:simplex-jensen} and get \(h(\vect z) \leq A(\vect z)\) for any affine function on \(\aff(P)\) agreeing with \(h\) at the vertices of \(P\).
    Note that
    \[h(\vect z) \leq A(\vect z) = A\left(\sum_{i=0}^r \bar \lambda^\star_i \vect p_i\right) = \sum_{i=0}^r \bar \lambda^\star_i A(\vect p_i)= \sum_{i=0}^r\bar{\lambda}^\star_i h(\vect p_i) \leq \sum_{i=0}^a \lambda_i h(\vect p_i),
    \]
    where the final inequality follows from LP optimality.
  \end{claimproof}

  \begin{claim}[Finite convex envelope]\label{clm:finite-envelope}
    The function $C:\R^k\to\R\cup\{-\infty\}$ defined by
    \begin{equation*}
      C(\vect z)=\inf\left\{
      \sum_{i=0}^a\lambda_i h(\vect p_i):
      \begin{array}{l}
        a\ge0,\ \vect p_i\in k\Z^k,\ \lambda_i\ge0, \\
        \sum_i\lambda_i=1,\ \sum_i\lambda_i \vect p_i=\vect z
      \end{array}\right\}
    \end{equation*}
    is everywhere finite and convex, and satisfies $C(\vect z)=h(\vect z)$ for
    every $\vect z\in k\Z^k$.
  \end{claim}
  \begin{claimproof}
    We first show that every real point has a representation of finite cost.
    Every $\vect x\in\R^k$ lies in a lattice cell with vertices in $k\Z^k$ and is a convex combination of those vertices.
    Since $h$ is finite at each vertex by the assumption of~\cref{lem:finite-function}, this representation has finite cost.
    Thus $C(\vect x)<+\infty$.

    Next, let $\vect z\in k\Z^k$.
    By \cref{clm:lattice-jensen}, every lattice representation of $\vect z$ has cost at least $h(\vect z)$.
    The one-point representation $\vect z=1\cdot\vect z$ has cost exactly $h(\vect z)$.
    Hence
    \[
      C(\vect z)=h(\vect z).
    \]

    We now establish a finite lower bound at an arbitrary point $\vect x\in\R^k$.
    Fix $\vect z_0\in k\Z^k$ and set $\vect q=2\vect z_0-\vect x$.
    Choose one lattice representation of $\vect q$, and denote its finite cost by $U$.

    Consider any lattice representation of $\vect x$, with cost $T$.
    Taking half of each weight in these two representations gives a lattice representation of
    \[
      \frac{\vect x+\vect q}{2}=\vect z_0
    \]
    with cost $(T+U)/2$.
    By \cref{clm:lattice-jensen},
    \[
      h(\vect z_0)\le\frac{T+U}{2},
      \qquad\text{so}\qquad
      T\ge2h(\vect z_0)-U.
    \]
    Thus, every lattice representation of $\vect x$ has cost \(T\ge2h(\vect z_0)-U\).
    Taking the infimum over all such representations yields
    \[
      C(\vect x) \geq 2h(\vect z_0)-U>-\infty.
    \]
    Therefore $C$ is everywhere finite.

    Finally, let $\vect x,\vect y\in\R^k$, $\lambda \in[0,1]$, and $\varepsilon>0$.
    Choose representations of $\vect x$ and $\vect y$ with costs at most $C(\vect x)+\varepsilon$ and $C(\vect y)+\varepsilon$, respectively.
    Multiplying their weights by $\lambda$ and $1-\lambda$ and combining them gives a representation of
    $\lambda\vect x+(1-\lambda)\vect y$. Consequently,
    \[
      C\bigl(\lambda\vect x+(1-\lambda)\vect y\bigr)
      \le \lambda C(\vect x)+(1-\lambda)C(\vect y)
      +\varepsilon.
    \]
    Letting $\varepsilon\to0$ proves convexity.
  \end{claimproof}

  We can now complete the proof of~\Cref{lem:finite-function}.
  By \cref{clm:finite-envelope}, $\vect v\mapsto C(k\vect v)$ is a finite convex extension of $\vect v\mapsto h(k\vect v)$.
  This proves the lemma for $\vect\tau=\mathbf{0}$; applying the same argument to $\vect u\mapsto h(\vect\tau+ \vect u)$ proves the statement for every translation $\vect\tau\in\Z^k$.
\end{proof}

\needspace{7\baselineskip}
\begin{corollary}[The optimization period]\label{cor:period}
  Assume $k\ge1$ and set $M=kN$. For every $\vect q\in\Z^k$, the function
  \begin{equation*}
    F_{\vect q}(\vect v)=F(\vect q+M\vect v),\qquad \vect v\in\Z^k,
  \end{equation*}
  is convex extensible.
\end{corollary}
\begin{proof}
  Choose $\vect r\in\{0,\ldots,N-1\}^k$ congruent to $\vect q$ modulo $N$, and
  let $\vect \tau=(\vect q-\vect r)/N\in\Z^k$. Then
  \begin{align*}
    F_{\vect q}(\vect v) & = F(\vect q + Nk\vect v) = F(\vect r + \vect q - \vect r + Nk\vect v) = F(\vect r + N((\vect q-\vect r)/N+k\vect v)) = F(\vect r + N(\vect \tau+k\vect v)) \\
                         & = H_{\vect r}(\vect \tau+k\vect v).
  \end{align*}
  By \cref{lem:periodic-midpoint,lem:finite-function} this function is convex extensible.
\end{proof}

\section{The Algorithm}\label{sec:algorithm}
For now, assume that all variables are bounded; we prove the generalization later.
Let $M \coloneqq kN$ as in~\Cref{cor:period}.

For a residue $\vect q\in\{0,\ldots,M-1\}^k$, substitute $\vect x=\vect q+M\vect v$.
Since \(\vect \ell^{(0)} \leq \vect x \leq \vect u^{(0)}\), the feasibility region for the vector \(\vect v\) is \(\vect \alpha_{\vect q} \leq \vect v \leq \vect \beta_{\vect q}\), with:
\begin{equation}\label{eq:phase-box}
  \vect \alpha_{\vect q}=\left\lceil\frac{\vect \ell^{(0)}-\vect q}{M}\right\rceil,
  \qquad
  \vect \beta_{\vect q}=\left\lfloor\frac{\vect u^{(0)}-\vect q}{M}\right\rfloor.
\end{equation}
The following consequence of \cite[Theorem~1]{VGZC20} gives the
optimization step for each residue class.

\begin{lemma}[Optimization within a residue class]\label{lem:box}
  Suppose that $N$ and $M=kN$ have been computed.
  Given a residue $\vect q\in\{0,\ldots,M-1\}^k$, one can compute a minimizer $\vect v \in \Z^k$ of
  \begin{equation}\label{eq:min-global}
    \min\{F_{\vect q}(\vect v): \vect\alpha_q \leq \vect v \leq \vect\beta_q\},
  \end{equation}
  its value, and a local vector $\vect y$ attaining the inner minimum in \eqref{eq:value} at $\vect x=\vect q+M\vect v$ or report infeasibility, in time
  \begin{equation*}
    f(k,\bar\Delta) \cdot (|I|+\log N)^{O(1)}.
  \end{equation*}
\end{lemma}
\begin{proof}
  Compute the endpoints from \eqref{eq:phase-box} and report an empty box if some coordinate satisfies $\alpha_{q,j}>\beta_{q,j}$.
  This takes polynomial time in $|I|+\log N$.
  Thus, assume that the box is nonempty.
  Our goal is to apply~\Cref{thm:veselov}, which is stated for convic functions on a ball-domain.
  For completeness, we first show how to reduce the box-domain to a ball-domain and then show that the function is convic.

  Write $\vect\alpha=\vect\alpha_q$ and $\vect\beta=\vect\beta_q$, translate by $\vect v=\vect\alpha+\vect w$, and set
  \begin{equation}\label{eq:def-J}
    \vect R=\vect\beta-\vect\alpha,\qquad
    \vect x_*=\vect q+M\vect\alpha,\qquad
    J(\vect w)=F(\vect x_*+M\vect w).
  \end{equation}
  Every \(k\)-dimensional integer vector \(0 \leq \vect w \leq \vect R\) lies in the Euclidean ball $B_\rho$ centered at the origin, where
  \[
    \rho=\max\{1,k\norm{\vect R}_\infty\}.
  \]
  For every $\vect w\in B_\rho$, the corresponding global vector has infinity norm at most $X=\norm{\vect x_*}_\infty+M\rho$.
  Define the following upper bounds:
  \begin{equation*}
    \begin{split}
      Y & =\max\{\norm{\ellloc}_\infty,\norm{\uloc}_\infty\}, \\
      T & =\norm{\vect c^{(0)}}_1X+\norm{\cloc}_1Y
      +K\left(\norm{\vect b}_1
      +X\sum_{i,j}|Z_{ij}|+Y\sum_{i,j}|E_{ij}|\right).
    \end{split}
  \end{equation*}
  Every local vector in its box has infinity norm at most $Y$.
  The triangle inequality therefore bounds the absolute value of every candidate objective in \eqref{eq:value} by $T$, including the global cost.
  Hence $|J(\vect w)|\le T$ for every $\vect w\in B_\rho\cap\Z^k$ (see~\eqref{eq:def-J}).
  All these quantities have polynomial encoding length in $|I|+\log N$. In particular, $\log(2+\rho)$ is polynomially bounded in $|I|+\log N$.

  Define on the whole integer lattice
  \begin{equation*}
    \widetilde J(\vect w)=J(\vect w)
    +(2T+1)\sum_{j=1}^k\dist(w_j,[0,R_j]),
  \end{equation*}
  where $\dist(t,[0,R_j])=\max\{0,-t,t-R_j\}$ is a convex function.
  By \cref{cor:period}, $J$ is convex extensible, and adding this convex distance penalty preserves convex extensibility.
  At an integer point outside the box, the sum of distances is at least one.
  Thus every such point in $B_\rho$ has penalized value at least $-T+(2T+1)=T+1$, whereas every point of the box has penalized value $J(\vect w)\le T$.
  Consequently, every minimizer of $\widetilde J$ over $B_\rho\cap\Z^k$ lies in the box and minimizes $J$ there if at least one feasible point exists.

  To see that convex extensibility implies discrete convicity, let $H:\R^k\to\R$ be a convex extension of $h:\Z^k\to\R$.
  Take integer points $\vect x_i,\vect y,\vect z$ satisfying $\vect z=\vect y+\sum_i\lambda_i(\vect y-\vect x_i)$, where $\lambda_i\ge0$ and $h(\vect x_i)\le h(\vect y)$.
  Put $\Lambda=\sum_i\lambda_i$.
  Then
  \begin{align*}
     & \vect z=\vect y+\sum_i\lambda_i(\vect y-\vect x_i) =\vect y+\sum_i\lambda_i\vect y-\sum_i\lambda_i\vect x_i \\
     & \Leftrightarrow \vect z + \sum_i\lambda_i\vect x_i=\vect y+\sum_i\lambda_i\vect y                           \\
     & \Leftrightarrow \vect z + \sum_i\lambda_i\vect x_i=(1 + \Lambda) \vect y                                    \\
     & \Leftrightarrow \vect y=\frac{\vect z+\sum_i\lambda_i\vect x_i}{1+\Lambda}.
  \end{align*}
  Convexity of $H$, which agrees with $h$ at these integer points, gives
  \[
    (1+\Lambda)h(\vect y)
    \le h(\vect z)+\sum_i\lambda_i h(\vect x_i)
    \le h(\vect z)+\Lambda h(\vect y).
  \]
  Subtracting $\Lambda h(\vect y)$ yields $h(\vect y)\le h(\vect z)$,
  precisely the defining inequality for a discrete convic function.

  Thus, the algorithm by~\cite{VGZC20} (see~\cref{thm:veselov}) applies to $\widetilde J$ on $B_\rho\cap\Z^k$.
  It uses $2^{O(k^2\log(k+1))}\cdot\log(2+\rho)$ comparison queries and $2^{O(k^2\log(k+1))}\cdot (\log(2+\rho))^{O(1)}$ arithmetic operations.
  Each comparison query between two points evaluates $\widetilde J$, which requires evaluating $F$ twice.
  By \cref{lem:oracle}, each evaluation of $F$ takes time $f(k,\bar\Delta)\cdot (|I| + \log N)^{O(1)}$.
  Multiplying the number of queries by this evaluation cost and adding the internal operations of the algorithm yields a running time of
  \begin{equation*}
    2^{O(k^2\log(k+1))}\,
    f(k,\bar\Delta) \cdot (|I|+\log N)^{O(1)}.
  \end{equation*}
  Let $\vect w^\star$ be the returned minimizer and set $\vect v=\vect\alpha_{\vect q}+\vect w^\star$.
  By the penalty argument, $\vect v$ minimizes \eqref{eq:min-global}.
  One final evaluation at $\vect x=\vect q+M\vect v$ returns the optimum value and an attaining local vector $\vect y$ by \cref{lem:oracle} within the same bound.
  Absorbing $2^{O(k^2\log(k+1))}f(k,\bar\Delta)$ into a computable function $f'(k,\bar\Delta)$ completes the proof.
\end{proof}

We now have all the ingredients to prove correctness and running time for
the finite-bound algorithm.

\begin{theorem}[Solving instances with finite bounds]\label{thm:bounded-time}
  Given finite integral lower and upper bounds, an instance $I$ of \eqref{eq:fourblock} can be solved in time $f(k,\bar\Delta)\cdot {|I|}^{O(1)}$.
  The algorithm returns an optimal integer solution or reports infeasibility.
\end{theorem}
\begin{proof}
  Compute \(W, K\) from~\eqref{eq:penalty-weight}.
  Choose \(N\) from~\cref{prop:idp} for block dimensions at most \(3k\) and entry bound \(\bar\Delta\), and set \(M=kN\).
  Every global integer vector \(\vect x \in \Z^k\) has a unique expression $\vect x=\vect q+M\vect v$ with $\vect q\in\{0,\ldots,M-1\}^k$.
  For a fixed $\vect q$, the vector lies in the global box exactly when $\vect \alpha_q\leq \vect v \leq \vect \beta_q$, where $\vect \alpha_q$ and $\vect \beta_q$ are defined in~\eqref{eq:phase-box}.
  Thus, the residue boxes partition the global search space, and at least one such box exists.

  For each residue vector $\vect q$, apply \cref{lem:box}.
  Discard an empty residue box; otherwise obtain a minimizer $\vect v$, its value $F_{\vect q}(\vect v)$, and an attaining local vector $\vect y$.
  Set $\vect x=\vect q+M\vect v$.
  Choose a pair $(\vect x^\star,\vect y^\star)$ with the smallest value among all residues \(\vect q\).
  This pair minimizes the penalized objective over the original product box.
  By \cref{lem:exact}, if $Z\vect x^\star+E\vect y^\star=\vect b$, it is an optimal solution of \eqref{eq:fourblock}; otherwise the instance is infeasible.
  This proves correctness.

  It remains to bound the running time. The algorithm inspects
  \begin{equation*}
    M^k=(kN)^k
  \end{equation*}
  residue classes.
  By \cref{lem:box}, all work for one residue class, including the emptiness check, value evaluations, and witness recovery, takes time \(f(k,\bar\Delta) \cdot (|I|+\log N)^{O(1)}\).
  Multiplying the number of residue classes by this bound gives a total running time of
  \[
    (kN)^k f(k,\bar\Delta)\cdot (|I|+\log N)^{O(1)}.
  \]
  Moreover, $N$ is bounded by a function in $(k,\bar\Delta)$.
  Absorbing this term, $(kN)^k$, and the dependence on $\log N$ into a computable parameter function gives the claimed running time.
\end{proof}

\subsection*{General Bounds and Unboundedness}
\label{sec:unbounded}

The following theorem removes the assumption of finite variable bounds
and completes the proof of \cref{thm:main}.

\mainthm*
\begin{proof}
  Normalize the instance \(I\) by shifting variables with finite lower bounds, reflecting variables with only a finite upper bound, and splitting free variables \(z_i\) into a difference of two nonnegative variables \(z_i^+-z_i^-\).
  Encode each remaining finite upper bound using a nonnegative slack, placing the slack and its equation in the same block as the original variable.
  The resulting instance has the form
  \begin{equation}\label{eq:normalized}
    \min\{\langle \widetilde{\vect c}, \vect z\rangle +\gamma:
    H\vect z=\vect b',\ \vect z\in\Z_{\ge0}^{d}\}.
  \end{equation}
  Each original variable gives at most two new variables, and each block gains at most $k$ rows.
  Thus the normalized instance has 4-block structure with block dimensions at most $2k$ and unchanged coefficient bounds $\Delta, \bar\Delta$.
  Reversing the substitutions recovers an original solution with the
  same objective value.

  Let $m$ be the total number of constraints.
  If $m=0$, the normalized instance is feasible: a negative coefficient of $\widetilde{\vect c}$ makes the objective unbounded, and otherwise $\vect z=0$ is optimal.
  Hence assume $m\ge1$ and set
  \begin{equation*}
    \beta=\max\{1,\norm{H}_{\max},\norm{\vect b'}_\infty\},\qquad
    \Gamma=(2m\beta+1)^m,
  \end{equation*}
  where $\norm{H}_{\max}$ is the largest absolute entry of $H$.

  For any nonnegative integer solution $\vect z$ to~\eqref{eq:normalized}, the vector $(\vect z,1)$ belongs to the integer kernel of $[H\ {-\vect b'}]$.
  By the conformal decomposition property~\cite{Cook1986AnIA}, \((\vect z, 1)\) is a sum of Graver elements conformal to $(\vect z,1)$.
  Since \(\vect z\) is nonnegative, all summands are therefore nonnegative.
  Since their last coordinates are integers summing to~$1$, exactly one summand has last coordinate~$1$.
  Write it as $(\bar{\vect z},1)$.
  All other summands have the form $(\vect g^i,0)$, with $\vect g^i\ge0$ and $H\vect g^i=0$.
  Thus $\vect z=\bar{\vect z}+\sum_i\vect g^i$ and
  \begin{align}\label{eq:bounded-problem}
    H\bar{\vect z}=\vect b',\qquad
    0\le\bar{\vect z}\le\vect z,\qquad
    \norm{\bar{\vect z}}_\infty\le\Gamma.
  \end{align}
  The matrix $[H\ {-\vect b'}]$ has $m$ rows and entries of absolute value at most $\beta$.
  Applying~\cite[Lemma~2]{EHK18} to its Graver element $(\bar{\vect z},1)$ gives
  \[
    \norm{(\bar{\vect z},1)}_1
    \le (2m\beta+1)^m=\Gamma.
  \]
  In particular, $\norm{\bar{\vect z}}_\infty\le\Gamma$.
  Thus, \eqref{eq:bounded-problem} preserves feasibility.
  Now solve the bounded feasibility problem~\eqref{eq:bounded-problem} by applying \cref{thm:bounded-time} with zero objective and upper bound \(\bar{\vect z} \leq \Gamma\mathds{1}\).
  If it is infeasible, report infeasibility.
  Otherwise, let $\vect z^0$ be a feasible solution.

  By \cite[Corollary~7.1b]{Sch86}, the nonempty polyhedron
  \[
    P=\{\vect z\ge0:H\vect z=\vect b'\}
  \]
  is the sum of a polytope \(Q\) and its recession cone
  \[
    C=\{\vect g\ge0:H\vect g=0\}.
  \]
  Now test in polynomial time~\cite{Kha80} if there exists some $\vect g\in C$ with $\langle\widetilde{\vect c},\vect g\rangle<0$.

  If such a rational direction \(\vect g \in C\) exists, multiply it by a common denominator to obtain an integer direction \(\vect g'\).
  The points $\vect z^0+t\vect g'$, $t\in\Z_{\ge0}$, then prove integer unboundedness of the system~\eqref{eq:normalized}.

  Otherwise, every nonnegative kernel vector has nonnegative cost.
  The decomposition $\vect z=\bar{\vect z}+\sum_i\vect g^i$ shows that every $\vect z$ that is a feasible solution to~\eqref{eq:normalized} has a representative $0 \leq \bar{\vect z} \leq \Gamma\mathds{1}$ with $\langle \widetilde{\vect c}, \bar{\vect z} \rangle \le\langle \widetilde{\vect c}, \vect z \rangle$.
  The finite, nonempty feasible set inside this box therefore contains a global integer optimum.
  Apply \cref{thm:bounded-time} again, now with objective $\langle \widetilde{\vect c}, \vect z \rangle$, and recover an optimal solution of the original instance.

  Finally, $\log\Gamma=m\log(2m\beta+1)$ is polynomial in $|I|$.
  The matrix $[H\ {-\vect b'}]$ is used only to prove this bound and is never passed to the finite-bound solver.
  Thus the algorithm uses at most two finite-bound calls, each with block dimensions at most $2k$, coefficient bound $\bar\Delta$ for the local-column blocks, and polynomial encoding length, together with one polynomial-time LP computation.
  Its total running time is $f(2k,\bar\Delta)\cdot |I|^{O(1)}$, which proves the theorem.
\end{proof}

\section*{Declaration of generative AI use}
During the preparation of this work, the authors used GPT-6 Astra and Gemini to assist in theoretical exploration and for discussions on the topic.
Our initial approach was to find a W[1]-hardness reduction.
While we were able to encode problems like subgraph isomorphism in 4-block ILPs, most reductions resulted in large coefficients in the block matrices.
This motivated us to look for FPT algorithms.
The final proof was found by GPT-6 Astra during a discussion with the authors. The first version was found by GPT-6 Astra on September 10, 2026. The solution we presented here mainly relies on a version found by GPT-6 Astra on September 16, 2026. The extension to separable convex objectives was also found during discussions with GPT-6 Astra on September 19, 2026 and was rewritten afterwards.
The authors independently verified the mathematical proofs for correctness, and improved the presentation.
The authors take full responsibility for the entire content, technical claims, and integrity of the paper.

\printbibliography

\appendix

\section{Omitted Proofs}
\label{sec:omitted-proofs}

\simplexMidpoints*
\begin{claimproof}
  Fix such a nonvertex $\vect w \in P \cap\Z^k$.
  Since the vertices are affinely independent, there are unique coefficients \(\beta_0,\ldots,\beta_r \in \R_{\geq 0}\) such that
  \begin{equation}\label{eq:scaled-barycentric}
    \vect w=\sum_{i=0}^r\beta_i \vect p_i,\qquad
    \beta_i\ge0,\qquad \sum_{i=0}^r\beta_i=k.
  \end{equation}
  If every \(\beta_i\) is integral, at least two coefficients are positive because \(\vect w\) is not a vertex.
  Choose \(i\ne j\) with \(\beta_i,\beta_j\ge1\), and set
  \[
    \vect u=\vect w-\vect p_i+\vect p_j,\qquad
    \vect v=\vect w+\vect p_i-\vect p_j.
  \]
  These transfers preserve nonnegativity and the sum of the coefficients.
  Thus \(\vect u,\vect v\) are distinct lattice points of \(P\) with midpoint \(\vect w\).

  Otherwise some $\beta_i$ is nonintegral. Since the sum of the fractional parts of the $r+1$ numbers $2\beta_i$ is an integer smaller than $r+1$,
  \[
    \sum_{i=0}^r\lfloor2\beta_i\rfloor\ge2k-r\ge k.
  \]
  Choose integers $a_i$ with $0\le a_i\le\lfloor2\beta_i\rfloor$ and $\sum_i a_i=k$.
  Such a choice exists because the integral upper bounds sum to at least $k$.
  Set
  \[
    \vect u=\sum_{i=0}^r a_i \vect p_i,\qquad \vect v=2\vect w-\vect u.
  \]
  The coefficients of $\vect u$ and of $\vect v=\sum_i(2\beta_i-a_i)\vect p_i$ are nonnegative and sum to $k$, so $\vect u,\vect v\in P$.
  Both are integral, because $a_i$ and $\vect p_i$ are integral for all $i$ and $\vect v=2\vect w-\vect u$ and $\vect w, \vect u$ are integral.
  If $\vect u=\vect v$, uniqueness of the representation~\eqref{eq:scaled-barycentric} would give $a_i=\beta_i$ for all $i$, contradicting the nonintegrality of some $\beta_i$. Thus \(\vect u,\vect v\) are distinct lattice points of \(P\) with midpoint \(\vect w\).
\end{claimproof}

\section{Separable Convex Objectives Given by Comparison Oracles}
\label{app:convex-comparison}

In this section, we extend the fixed-parameter tractability result for 4-block integer programming to separable convex objectives given by comparison oracles.
We consider instances with finite integral variable bounds.
This restriction is needed because distinguishing an attained minimum from an objective unbounded below is undecidable (see~\cite[Section~3.5.2, Proposition~60]{EHK+22}).

Kouteck\'y, Lassota, and Ligthart~\cite[Theorem~2]{KLL26} establish this extension under a more general coefficient parameterization.
Here we give a proof within our framework, assuming access to a comparison oracle.

Our proof uses the periodic convexity argument
from \Cref{sec:convexity}, the algorithm for minimizing discrete
convic functions in \Cref{thm:veselov}, and the
comparison-oracle algorithm for separable convex
integer programming of Hunkenschröder, Koutecký,
Levin, and Vu~\cite{HKLV26}.

\begin{theorem}[{\cite[Theorem 2]{HKLV26}}]\label{thm:hunkenschroeder}
  There is a computable function $g$ and an algorithm which solves generalized $n$-fold IPs with block dimension $k$, finite variable bounds and a separable convex function $f$ given by a comparison oracle as objective $\min\{f(\vect x)\ |\ A\vect x = \vect b, \vect \ell \leq \vect x \leq \vect u, \vect x \in \mathbb{Z}^{nk}\}$ in time \(g(k,\norm{A}_{\max})\cdot|I|^{O(1)}\).
\end{theorem}

The main difference from the linear case is that the numerical variation of the objective is not available.
We overcome this difficulty by implementing sufficiently large penalties implicitly through lexicographic comparisons.
The penalty coefficients are used only in the correctness proof and are never computed.

\subsection{Problem and Main Result}

Consider the 4-block integer program with new objective function
\begin{equation}
  \label{eq:sc-4block}\tag{4-Block IP}
  \min\quad
  \Phi(\vect x, \vect y)= \Phi^{(0)}(\vect x)+\sum_{i=1}^n\Phi^{(i)}(\vect y^{(i)})\\
\end{equation}
where
\[
  \Phi^{(0)}(\vect x)=\sum_{j=1}^k\phi^{(0)}_{j}(x_j),
  \qquad
  \Phi^{(i)}(\vect y^{(i)})
  =\sum_{j=1}^k\phi^{(i)}_{j}(y^{(i)}_j),
\]
where every $\phi^{(i)}_{j}:\mathbb R\to\mathbb R$ is convex and finite-valued.

We assume access to a comparison oracle $\mathcal O_\Phi$ which, given two vectors $\vect z,\vect z'\in\mathbb Z^{(n+1)k}$, determines whether
\[
  \Phi(\vect z)<\Phi(\vect z'),\qquad
  \Phi(\vect z)=\Phi(\vect z'),\qquad\text{or}\qquad
  \Phi(\vect z)>\Phi(\vect z')
\]
in time $(|I| +  \enc(\vect z,\vect z'))^{O(1)}$.
Equivalently, one may use the binary comparison oracle of~\cite{HKLV26};
two calls suffice to distinguish the three outcomes.

The oracle is assumed to be available also at the integer points outside the original variable bounds that are queried by the algorithm.
This is consistent with the assumption that every component function is defined on all of $\mathbb R$.

\begin{theorem}[Separable convex 4-block IP]
  \label{thm:sc-4block-comparison}
  Let $I$ be an instance of \eqref{eq:sc-4block} with finite integral variable bounds and a separable convex objective given by a comparison oracle.

  There is a computable function $f$ such that $I$ can be solved in time
  \[
    f(k,\bar\Delta)\cdot |I|^{O(1)},
  \]
  where $|I|$ denotes the binary encoding length of $I$.

  The algorithm uses only comparisons of objective values.
  In particular, it requires neither numerical objective evaluations nor bounds on the magnitudes or differences of objective values.
\end{theorem}

We prove the theorem through three lemmas.
First we construct a generalized $n$-fold oracle which lexicographically minimizes constraint violations and convex costs.
Next we establish periodic convexity for an auxiliary value function.
Finally, we replace \Cref{lem:box} by a comparison-oracle version.

Throughout the proof, we use the matrix notation
\[
  Z\vect x+E\vect y=\vect b
\]
from \Cref{sec:introduction}, and write
\[
  Y=
  \{\vect y\in\mathbb Z^{nk}:
  \vect \ell^{(1:n)}\le \vect y\le \vect u^{(1:n)}\}.
\]
We may assume that all variable intervals are nonempty; otherwise infeasibility is
detected immediately.
In particular, $Y$ is a finite nonempty set.

\subsection{A Lexicographic Generalized
  $n$-Fold Oracle}

The first step is to replace the numerical exact penalty of \Cref{sec:penalty} by a lexicographic objective.

For $\vect x\in\mathbb Z^k$ define
\begin{equation}
  \label{eq:sc-residual}
  r(\vect x)=\min_{\vect y\in Y}\|\vect b-Z\vect x-E\vect y\|_1.
\end{equation}
Among the vectors attaining this minimum, choose for any given $\vect x \in \mathbb{Z}^k$
\begin{equation}
  \label{eq:sc-local-witness}
  \vect y_{\vect x}\in
  \operatorname*{argmin}
  \left\{
  \Phi(\vect x,\vect y):
  \vect y\in Y,\
  \|\vect b-Z\vect x-E\vect y\|_1=r(\vect x)
  \right\}.
\end{equation}
The vector $\vect y_{\vect x}$ exists since $Y$ is finite and nonempty.

We aim to lexicographically minimize pairs
\[
  (r(\vect x),\Phi(\vect x,\vect y_{\vect x}))
\]
by minimizing the residual $r(\vect x)$ first and the convex objective $\Phi(\vect x, \vect y_{\vect x})$ second.

\begin{lemma}[Lexicographic value-and-witness
    oracle]
  \label{lem:sc-local-oracle}
  For every $\vect x\in\mathbb Z^k$, the integer $r(\vect x)$ and a vector $\vect y_{\vect x}$ satisfying \eqref{eq:sc-local-witness} can be computed
  in time
  \[
    g(k,\bar\Delta)\cdot
    (|I|+\enc(\vect x))^{O(1)},
  \]
  where $g$ is a computable function.

  Moreover, for any $\vect x,\vect x'\in\mathbb Z^k$, the lexicographic pairs
  \[
    (r(\vect x),\Phi(\vect x,\vect y_{\vect x}))
    \quad\text{and}\quad
    (r(\vect x'),\Phi(\vect x',\vect y_{\vect x'}))
  \]
  can be compared in time \(g(k,\bar\Delta)\cdot
    (|I|+\enc(\vect x,\vect x'))^{O(1)}\), without evaluating either objective value.
\end{lemma}

\begin{proof}
  We solve two generalized $n$-fold integer programs.
  The first minimizes the integral residual.
  The second minimizes the separable convex objective subject to fixing the residual at its minimum.
  Only the second program needs the comparison oracle.

  Fix $\vect x\in\mathbb Z^k$.
  As in the proof of \Cref{lem:oracle}, define
  \[
    s^{(0)}(\vect x)=\vect b-Z\vect x-E\vect \ell^{(1:n)}
  \]
  and
  \[
    S(\vect x)=|s^{(0)}(\vect x)|
    +|E|(\vect u^{(1:n)}-\vect \ell^{(1:n)}).
  \]
  The absolute value of every residual $\vect b-Z\vect x-E\vect y$, for $\vect y\in Y$, is bounded componentwise by $S(\vect x)$.

  Consider the integer program
  \begin{equation}
    \label{eq:sc-first-oracle}
    \begin{aligned}
      \min\quad        &
      \mathds1^\top(\vect s^++\vect s^-)                \\
      \text{s.t.}\quad &
      E\vect y+\vect s^+-\vect s^-=\vect b-Z\vect x,    \\
                       &
      \vect \ell^{(1:n)}\le \vect y\le \vect u^{(1:n)}, \\
                       &
      \mathbf{0}\le \vect s^+,\vect s^-\le S(\vect x),  \\
                       &
      \vect y,\vect s^+,\vect s^-\text{ integral}.
    \end{aligned}
  \end{equation}
  For every $\vect y\in Y$, its residual can be represented by its positive and negative parts within the indicated bounds.
  Conversely, any feasible $(\vect y,\vect s^+,\vect s^-)$ satisfies
  \[
    \mathds1^\top(\vect s^++\vect s^-)
    \ge \|\vect b-Z\vect x-E\vect y\|_1.
  \]
  Consequently, the optimum of \eqref{eq:sc-first-oracle} equals $r(\vect x)$.

  This is a linear generalized $n$-fold program.
  The local residual variables are placed in their corresponding bricks and the residual variables of the global equations form one additional brick, as in the proof of \Cref{lem:oracle}.
  Each brick has at most $3k$ variables, there are $k$ global rows and at most $k$ local rows per brick, and all coefficients have absolute value at most $\bar\Delta$.
  Thus \Cref{thm:nfold} computes $r(\vect x)$ in fixed-parameter time.

  Now solve
  \begin{equation}
    \label{eq:sc-second-oracle}
    \begin{aligned}
      \min\quad        &
      \Phi(\vect x,\vect y)                             \\
      \text{s.t.}\quad &
      E\vect y+\vect s^+-\vect s^-=\vect b-Z\vect x,    \\
                       &
      \mathds1^\top(\vect s^++\vect s^-)=r(\vect x),    \\
                       &
      \vect \ell^{(1:n)}\le \vect y\le \vect u^{(1:n)}, \\
                       &
      \mathbf{0}\le \vect s^+,\vect s^-\le S(\vect x),  \\
                       &
      \vect y,\vect s^+,\vect s^-\text{ integral}.
    \end{aligned}
  \end{equation}

  Every feasible solution of \eqref{eq:sc-second-oracle} has a local vector $\vect y$ with
  \[
    \|\vect b-Z\vect x-E\vect y\|_1=r(\vect x).
  \]
  Indeed, its residual norm is at most $\mathds1^\top(\vect s^++\vect s^-)=r(\vect x)$, while $r(\vect x)$ is the minimum possible norm.
  Conversely, every local vector achieving $r(\vect x)$ admits such a representation.
  Therefore an optimal local vector of \eqref{eq:sc-second-oracle} is precisely a vector $\vect y_{\vect x}$ as in \eqref{eq:sc-local-witness}.

  The additional equality $\mathds1^\top(\vect s^++\vect s^-)=r(\vect x)$ is one additional global row.
  Hence \eqref{eq:sc-second-oracle} is still a generalized $n$-fold program, with at most $k+1$ global rows, at most $k$ local rows per brick, and at most $3k$ variables per brick.
  Its coefficient bound remains $\bar\Delta$.

  Since $\vect x$ is fixed, the objective $\Phi(\vect x,\vect y)$ is separable convex in the variables $\vect y,\vect s^+,\vect s^-$:
  it has the original convex component functions on the $\vect y$-coordinates and zero functions on the slack coordinates.

  Its comparison oracle is implemented by the original oracle $\mathcal O_\Phi$, keeping the global coordinates fixed at $\vect x$. Therefore \Cref{thm:hunkenschroeder}, applied to the augmented generalized $n$-fold matrix, returns an optimal solution of \eqref{eq:sc-second-oracle} in fixed-parameter time.

  The entries of $S(\vect x)$ and the right-hand sides of both programs have binary encoding length polynomial in $|I|+\enc(\vect x)$.
  Hence both programs are solved within the claimed running time.

  Finally, to compare $(r(\vect x),\Phi(\vect x,\vect y_{\vect x}))$ and $(r(\vect x'),\Phi(\vect x',\vect y_{\vect x'}))$, first compare the explicitly computed integers $r(\vect x),r(\vect x')$.
  If they are equal, compare
  \[
    \Phi(\vect x,\vect y_{\vect x})
    \quad\text{and}\quad
    \Phi(\vect x',\vect y_{\vect x'})
  \]
  using the original comparison oracle.
  No numerical evaluation is required.
\end{proof}

\subsection{Periodic Convexity for
  Separable Convex Objectives}

We next establish the structural property required for optimizing over the global variables.

Although the algorithm will not compute numerical penalty coefficients, it is useful to introduce such coefficients for the mathematical analysis.

For an arbitrary real number $K>0$, define the value function
\begin{equation}
  \label{eq:sc-penalized-value}
  F_K(\vect x)=
  \min_{\vect y\in Y}
  \left\{
  \Phi(\vect x,\vect y)+
  K\|\vect b-Z\vect x-E\vect y\|_1
  \right\},
  \qquad \vect x\in\mathbb Z^k.
\end{equation}
This function is everywhere finite because $Y$ is finite and nonempty.

Importantly, $K$ is fixed in the following structural argument but may be chosen arbitrarily.

Let
\[
  Q=[E\ I]
\]
be the generalized $n$-fold matrix from \Cref{sec:convexity}, with the residual variables grouped into bricks.

Let $N$ be the positive integer provided by \Cref{prop:idp} for $Q$.
The value of $N$ is bounded by a function in $k$ and $\bar\Delta$ alone, independently of the choice of $K$ and of the objective function.

We first recall the following elementary inequality~\cite{Marshall2011,Kadelburg2005}.

\begin{lemma}[Theorem 1 of \cite{Kadelburg2005}]
  \label{lem:sc-two-point}
  Let $\phi:\mathbb R\to\mathbb R$ be convex.
  Suppose $a,b,c,d\in\mathbb R$ satisfy
  \[
    a+b=c+d,\qquad
    \min\{a,b\}\le c,d\le\max\{a,b\}.
  \]
  Then
  \[
    \phi(c)+\phi(d)
    \le\phi(a)+\phi(b).
  \]
\end{lemma}

\begin{lemma}[Periodic midpoint convexity]
  \label{lem:sc-periodic-midpoint}
  For every $K>0$ and every $\vect r\in\{0,\ldots,N-1\}^k$, the function
  \[
    H_{K,\vect r}(\vect v)=F_K(\vect r+N\vect v),
    \qquad \vect v\in\mathbb Z^k,
  \]
  is everywhere finite and midpoint-convex.
\end{lemma}

\begin{proof}
  Finiteness has already been observed. We prove this lemma analogously to \Cref{lem:periodic-midpoint}.

  Take $\vect a',\vect a''\in\mathbb Z^k$ with integer midpoint
  \[
    \vect w=\frac{\vect a'+\vect a''}{2}.
  \]
  Write
  \[
    \vect x'=\vect r+N\vect a',\qquad
    \vect x''=\vect r+N\vect a'',\qquad
    \bar{\vect x}=\vect r+N\vect w.
  \]
  Choose optimal local witnesses $\vect y',\vect y''\in Y$ in \eqref{eq:sc-penalized-value} for $\vect x'$ and $\vect x''$, respectively.

  Set
  \[
    \vect s'=\vect b-Z\vect x'-E\vect y',\qquad
    \vect s''=\vect b-Z\vect x''-E\vect y''
  \]
  and
  \[
    \vect z'=(\vect y',\vect s'),\qquad \vect z''=(\vect y'',\vect s'').
  \]
  Then
  \[
    Q\vect z'=\vect b-Z\vect x',\qquad
    Q\vect z''=\vect b-Z\vect x''.
  \]
  For $\vect \delta=\vect z''-\vect z'$ we obtain
  \[
    Q\vect \delta
    =-NZ(\vect a''-\vect a')
    =2N\left(-Z\frac{\vect a''-\vect a'}2\right).
  \]
  The vector $(\vect a''-\vect a')/2$ is integral.
  Thus \Cref{cor:bisection} yields
  \[
    \vect \delta=\vect \delta'+\vect \delta'',\qquad
    \vect \delta',\vect \delta''\sqsubseteq\vect \delta,
  \]
  such that
  \[
    Q\vect \delta'=Q\vect \delta''
    =-NZ\frac{\vect a''-\vect a'}2.
  \]

  Define
  \[
    \widehat{\vect z}'=\vect z'+\vect \delta',
    \qquad
    \widehat{\vect z}''=\vect z'+\vect \delta''.
  \]
  Both new vectors satisfy
  \[
    Q\widehat{\vect z}'=Q\widehat{\vect z}''
    =\vect b-Z\bar{\vect x}.
  \]
  Write
  \[
    \widehat{\vect z}'=(\widehat{\vect y}',\widehat{\vect s}'),
    \qquad
    \widehat{\vect z}''=(\widehat{\vect y}'',\widehat{\vect s}'').
  \]

  Conformality implies that every coordinate of $\widehat{\vect z}'$ and $\widehat{\vect z}''$ lies between the corresponding coordinates of $\vect z'$ and $\vect z''$.
  Consequently,
  \[
    \widehat{\vect y}',\widehat{\vect y}''\in Y.
  \]
  Furthermore,
  \[
    \widehat{\vect z}'+\widehat{\vect z}''
    =2\vect z'+\vect \delta'+\vect \delta''
    =\vect z'+\vect z''.
  \]

  Applying \Cref{lem:sc-two-point} to each coordinate of the local variables and summing yields
  \begin{align*}
     & \sum_{i=1}^n
    \Phi^{(i)}(\widehat{\vect y}'^{(i)})
    +\sum_{i=1}^n
    \Phi^{(i)}(\widehat{\vect y}''^{(i)}) \\
     & \qquad\le
    \sum_{i=1}^n\Phi^{(i)}(\vect y'^{(i)})
    +\sum_{i=1}^n\Phi^{(i)}(\vect y''^{(i)}).
  \end{align*}

  The same lemma applies to the convex function $t\mapsto |t|$ on every residual coordinate.
  Hence
  \[
    \|\widehat{\vect s}'\|_1+
    \|\widehat{\vect s}''\|_1
    \le\|\vect s'\|_1+\|\vect s''\|_1.
  \]

  Finally, convexity of the global objective gives
  \[
    2\Phi^{(0)}(\bar{\vect x}) =2\Phi^{(0)}(\vect r+N\vect w) = 2\Phi^{(0)}(\vect r+N\frac{\vect a' + \vect a''}{2}) = 2\Phi^{(0)}(\frac{\vect r + N \vect a'}{2}+\frac{\vect r + N\vect a''}{2})
    \le\Phi^{(0)}(\vect x')+\Phi^{(0)}(\vect x'').
  \]
  Using the constructed witnesses for the two copies of $F_K(\bar{\vect x})$ and adding the preceding inequalities gives
  \[
    2F_K(\bar{\vect x})
    \le F_K(\vect x')+F_K(\vect x'').
  \]
  Equivalently,
  \[
    2H_{K,\vect r}(\vect w)
    \le H_{K,\vect r}(\vect a')+H_{K,\vect r}(\vect a'').
  \]
  Thus $H_{K,\vect r}$ is midpoint-convex.
\end{proof}

\begin{corollary}[Convex extensibility on residue classes]
  \label{cor:sc-convex-extension}
  Let $M=kN$. For every $K>0$ and $\vect q\in\mathbb Z^k$, the function
  \[
    \vect v\longmapsto F_K(\vect q+M\vect v),
    \qquad \vect v\in\mathbb Z^k,
  \]
  is convex extensible.
\end{corollary}

\begin{proof}
  Choose $\vect r\in\{0,\ldots,N-1\}^k$ with $\vect q\equiv \vect r\pmod N$ and write $\vect q=\vect r+N\vect \tau$ for $\vect \tau\in\mathbb Z^k$.

  Then
  \[
    F_K(\vect q+M\vect v)
    =F_K(\vect r+N(\vect \tau+k\vect v))
    =H_{K,\vect r}(\vect \tau+k\vect v).
  \]
  By Lemma~\ref{lem:sc-periodic-midpoint}, $H_{K,\vect r}$ is everywhere finite and midpoint-convex.
  \Cref{lem:finite-function} therefore shows that the last function is convex extensible.
\end{proof}

\subsection{Implicit Exact Penalties and Optimization Within a Residue Class}

We now establish the central algorithmic step.
It replaces \Cref{lem:box} for separable convex objectives given by comparison oracles.

For a nonempty residue box, the algorithm will compare triples consisting of the distance to the box, the minimum residual, and the original convex cost.

The difficulty is that the algorithm of \Cref{thm:veselov} expects a comparison oracle for a real-valued discrete convic function, not an abstract lexicographically ordered tuple.

The next proof resolves this issue by showing that the lexicographic comparisons agree exactly with those of a real-valued convex-extensible function on the finite search ball.

\begin{lemma}[Comparison-oracle version of \Cref{lem:box}]
  \label{lem:sc-residue-optimization}
  Let $M=kN$, and let $\vect q\in\{0,\ldots,M-1\}^k$ be a residue.

  In time
  \[
    f_1(k,\bar\Delta) \cdot
    (|I|+\log N)^{O(1)}
  \]
  one can either determine that there is no global variable $\vect \ell^{(0)} \leq \vect x \leq \vect u^{(0)}$ congruent to $\vect q$ modulo $M$, or find a pair $(\vect x_{\vect q},\vect y_{\vect q})$ with
  \[
    \vect x_{\vect q}\equiv \vect q\pmod M,\qquad
    \vect \ell^{(0)}\le \vect x_{\vect q}\le \vect u^{(0)},
    \qquad \vect y_{\vect q}\in Y,
  \]
  which lexicographically minimizes
  \begin{equation}
    \label{eq:sc-lex-pair}
    \left(
    \|\vect b-Z\vect x-E\vect y\|_1,\,
    \Phi(\vect x,\vect y)
    \right)
  \end{equation}
  over all pairs in this residue class.

  The algorithm uses only the comparison oracle for $\Phi$ and never evaluates the numerical value of $\Phi$.
\end{lemma}

\begin{proof}
  First, we give a short overview of the proof idea:
  We embed the residue box into an integer ball, just as in \Cref{lem:box}.
  On this finite ball we prove the existence of two sufficiently large penalty coefficients.
  The first enforces minimum constraint violation, while the second enforces membership in the original residue box.
  We never compute these coefficients.
  Instead, their exact effect is reproduced by lexicographically comparing three quantities.
  The auxiliary numerical function defined using the existential coefficients is convex extensible.
  Thus the resulting comparison oracle is a valid oracle for \Cref{thm:veselov}.

  As in (\ref{eq:phase-box}), define
  \[
    \vect \alpha_{\vect q}=
    \left\lceil\frac{\vect \ell^{(0)}-\vect q}{M}\right\rceil,
    \qquad
    \vect \beta_{\vect q}=
    \left\lfloor\frac{\vect u^{(0)}-\vect q}{M}\right\rfloor.
  \]
  If some coordinate satisfies $\alpha_{\vect q,j}>\beta_{\vect q,j}$, the box is empty and we report this fact.

  Otherwise put
  \[
    \vect \alpha=\vect \alpha_{\vect q},\qquad
    \vect \beta=\vect \beta_{\vect q},\qquad
    \vect R=\vect \beta-\vect \alpha\ge\mathbf{0}
  \]
  and
  \[
    \vect x_*=\vect q+M\vect \alpha.
  \]
  Every global vector in the residue class has the form
  \[
    \vect x=\vect x_*+M\vect w,\qquad
    \vect w\in[\mathbf{0},\vect R]\cap\mathbb Z^k.
  \]

  Set
  \[
    \rho=\max\{1,k\|\vect R\|_\infty\}.
  \]
  Then
  \[
    [\mathbf{0},\vect R]\cap\mathbb Z^k
    \subseteq B_\rho\cap\mathbb Z^k.
  \]
  Denote the finite set of global vectors represented by the ball by
  \begin{equation}
    \label{eq:sc-X}
    X=\{\vect x_*+M\vect w:
    \vect w\in B_\rho\cap\mathbb Z^k\}.
  \end{equation}

  For $\vect w\in\mathbb Z^k$ define
  \begin{equation}
    \label{eq:sc-distance}
    D(\vect w)=
    \sum_{j=1}^k
    \operatorname{dist}(w_j,[0,R_j]),
  \end{equation}
  where
  \[
    \operatorname{dist}(t,[0,R_j])
    =\max\{0,-t,t-R_j\}.
  \]
  This function is convex on $\mathbb R^k$, takes nonnegative integer values at integer points, and satisfies
  \[
    D(\vect w)=0
    \quad\Longleftrightarrow\quad
    \vect w\in[\mathbf{0},\vect R]\cap\mathbb Z^k
    \qquad(\vect w\in\mathbb Z^k).
  \]

  Since $X$ and $Y$ are finite and $\Phi$ is finite-valued, the real number
  \[
    W= \max_{(\vect x,\vect y)\in X\times Y}\Phi(\vect x,\vect y) - \min_{(\vect x,\vect y)\in X\times Y}\Phi(\vect x,\vect y)
  \]
  is finite and nonnegative.

  Choose any real number
  \[
    K>W.
  \]
  This choice is made only in the mathematical proof.
  The algorithm does not need to know $W$ or $K$.

  Consider the function $F_K$ defined in~\eqref{eq:sc-penalized-value}.

  We claim that, for every $\vect x\in X$,
  \begin{equation}
    \label{eq:sc-implicit-penalty}
    F_K(\vect x)=Kr(\vect x)+\Phi(\vect x,\vect y_{\vect x}).
  \end{equation}

  To prove the claim, fix $\vect x\in X$.
  The vector $\vect y_{\vect x}$ attains residual $r(\vect x)$, so it gives the candidate value
  \[
    Kr(\vect x)+\Phi(\vect x,\vect y_{\vect x}).
  \]

  On the other hand, suppose $\vect y\in Y$ has residual strictly larger than $r(\vect x)$.
  Since all constraint data and variables are integral, the residual norms are integers and therefore
  \[
    \|\vect b-Z\vect x-E\vect y\|_1\ge r(\vect x)+1.
  \]
  By the definition of $W$,
  \[
    \Phi(\vect x,\vect y)
    \ge \Phi(\vect x,\vect y_{\vect x})-W.
  \]
  Consequently,
  \begin{align*}
     & \Phi(\vect x,\vect y)+K\|\vect b-Z\vect x-E\vect y\|_1 \\
     & \quad\ge
    \Phi(\vect x,\vect y_{\vect x})-W+K(r(\vect x)+1)         \\
     & \quad>
    \Phi(\vect x,\vect y_{\vect x})+Kr(\vect x).
  \end{align*}
  Thus no vector with larger residual can minimize $F_K(\vect x)$.

  Among vectors with residual $r(\vect x)$, the vector $\vect y_{\vect x}$ minimizes $\Phi$.
  This proves~\eqref{eq:sc-implicit-penalty}.

  In particular, for any $\vect x,\vect x'\in X$, comparison of $F_K(\vect x)$ and $F_K(\vect x')$ agrees exactly with lexicographic comparison of
  \begin{equation}
    \label{eq:compare}
    (r(\vect x),\Phi(\vect x,\vect y_{\vect x}))
    \quad\text{and}\quad
    (r(\vect x'),\Phi(\vect x',\vect y_{\vect x'})).
  \end{equation}
  Indeed, if the integer residuals differ, their difference is at least one in absolute value, while the difference between the two original objective values is at most $W<K$.

  If the residuals agree, the comparison is simply the original objective comparison.

  The set $X$ is finite, so the variation
  \[
    V_K=
    \max_{\vect x\in X}F_K(\vect x)
    -
    \min_{\vect x\in X}F_K(\vect x)
  \]
  is finite.

  Choose an arbitrary real number
  \[
    L>V_K.
  \]
  Again, $L$ exists but is never computed by the algorithm.

  Define a real-valued function on the entire integer lattice by
  \begin{equation}
    \label{eq:sc-aux-function}
    \widetilde J(\vect w)
    =F_K(\vect x_*+M\vect w)+LD(\vect w),
    \qquad \vect w\in\mathbb Z^k.
  \end{equation}

  We show that its comparisons on $B_\rho\cap\mathbb Z^k$ agree with lexicographic comparisons of
  \begin{equation}
    \label{eq:sc-lex-key}
    \mathcal L(\vect w)=
    \left(
    D(\vect w),\
    r(\vect x_*+M\vect w),\
    \Phi(\vect x_*+M\vect w,\vect y_{\vect x_*+M\vect w})
    \right).
  \end{equation}

  Let $\vect w,\vect w'\in B_\rho\cap\mathbb Z^k$.

  First suppose $D(\vect w)<D(\vect w')$.
  Since $D$ is integer-valued,
  \[
    D(\vect w')-D(\vect w)\ge1.
  \]
  The definition of $V_K$ implies
  \begin{align*}
    \widetilde J(\vect w')
    -\widetilde J(\vect w)
     & =
    F_K(\vect x_*+M\vect w')-F_K(\vect x_*+M\vect w) \\
     & \quad +L(D(\vect w')-D(\vect w))              \\
     & \ge -V_K+L                                    \\
     & >0.
  \end{align*}
  Thus $\widetilde J(\vect w)<\widetilde J(\vect w')$, independently of the residuals or convex objective values.

  The reverse implication holds symmetrically if $D(\vect w)>D(\vect w')$.

  If $D(\vect w)=D(\vect w')$, the box penalty cancels, and comparison of $\widetilde J(\vect w)$ and $\widetilde J(\vect w')$ reduces to comparison of the corresponding values of $F_K$.

  By (\ref{eq:compare}), this agrees with lexicographic comparison of the residuals and original objective values.

  Therefore, for every pair of integer points in the ball,
  \begin{equation}
    \label{eq:jcompl}
    \widetilde J(\vect w)\le\widetilde J(\vect w')
    \quad\Longleftrightarrow\quad
    \mathcal L(\vect w)\le_{\mathrm lex}
    \mathcal L(\vect w').
  \end{equation}
  Equality and strict inequalities are preserved as well.

  This is the central property:
  although neither $K$ nor $L$ is known, the comparisons of $\widetilde J$ can be implemented exactly.

  We now establish that $\widetilde J$ is convex extensible, and hence discrete convic.

  By \Cref{cor:sc-convex-extension}, the function
  \[
    \vect w\longmapsto F_K(\vect x_*+M\vect w)
  \]
  is convex extensible for every $K>0$, since $\vect x_*$ is an integer vector and $M=kN$.

  Thus there exists an everywhere-finite convex function $G:\mathbb R^k\to\mathbb R$
  such that
  \[
    G(\vect w)=F_K(\vect x_*+M\vect w)
    \qquad(\vect w\in\mathbb Z^k).
  \]

  The distance function $D$ extends to the convex function
  \[
    D(\vect z)=\sum_{j=1}^k
    \max\{0,-z_j,z_j-R_j\},
    \qquad \vect z\in\mathbb R^k.
  \]
  Since $L>0$, the function
  \[
    G(\vect z)+LD(\vect z)
  \]
  is convex and everywhere finite
  on $\mathbb R^k$.

  At every integer point it agrees with $\widetilde J$.
  Consequently, $\widetilde J$ is convex extensible.
  Since convex extensibility implies discrete convicity, the restriction of $\widetilde J$ to $B_\rho\cap\mathbb Z^k$ is discrete convic.

  Now, we can give the implementation of the comparison oracle.
  Given $\vect w,\vect w'\in B_\rho\cap\mathbb Z^k$, compute
  \[
    \vect x=\vect x_*+M\vect w,
    \qquad        \vect x'=\vect x_*+M\vect w'.
  \]
  By \Cref{lem:sc-local-oracle}, compute $r(\vect x),\vect y_{\vect x}$ and $r(\vect x'),\vect y_{\vect x'}$.

  Compute $D(\vect w)$ and $D(\vect w')$ using~\eqref{eq:sc-distance}.

  Now compare the keys
  \[
    \mathcal L(\vect w)
    =(D(\vect w),r(\vect x),\Phi(\vect x,\vect y_{\vect x}))
  \]
  and
  \[
    \mathcal L(\vect w')
    =(D(\vect w'),r(\vect x'),\Phi(\vect x',\vect y_{\vect x'}))
  \]
  lexicographically.

  The first two components are explicitly available integers.
  The third components are compared using the original oracle $\mathcal O_\Phi$.

  By (\ref{eq:jcompl}), the resulting answer is exactly the comparison of $\widetilde J(\vect w)$ and $\widetilde J(\vect w')$.

  Thus we have implemented a valid comparison oracle for the real-valued discrete convic function $\widetilde J$, without computing $K$, $L$, $F_K$, or any numerical objective value.

  We may now apply \Cref{thm:veselov} to minimize $\widetilde J$ over $B_\rho\cap\mathbb Z^k$.

  The algorithm uses
  \[
    2^{O(k^2\log(k+1))}
    \log(2+\rho)
  \]
  comparison queries and
  \[
    2^{O(k^2\log(k+1))}
      (\log(2+\rho))^{O(1)}
  \]
  additional arithmetic operations.

  Each comparison query is answered by the procedure described above.
  \Cref{lem:sc-local-oracle} shows that its running time is fixed-parameter in $k,\bar\Delta$, provided the queried global vectors have polynomial encoding length.

  We verify this last property.
  For every $\vect w\in B_\rho$,
  \[
    \|\vect x_*+M\vect w\|_\infty
    \le\|\vect x_*\|_\infty+M\rho.
  \]
  The definitions of $\vect x_*$, $\vect R$, $\rho$, and $M=kN$ show that
  \[
    \log(2+\|\vect x_*\|_\infty+M\rho)
    \le (|I|+\log N+\log(k+1))^{O(1)}.
  \]
  In particular, every global vector queried by the algorithm has binary encoding length polynomial in $|I|+\log N$, with parameter dependence absorbed into $f_1$.

  Therefore the complete minimization over the ball takes time
  \[
    f_1(k,\bar\Delta)\cdot (|I|+\log N)^{O(1)}.
  \]

  Finally, we show how to recover an optimum in the original residue box.
  Let $\vect w^\star$ be a minimizer returned by \Cref{thm:veselov}.

  Since the original integer box $[\mathbf{0},\vect R]$ is nonempty, there exists an integer point $\vect w^{(0)}\in[\mathbf{0},\vect R]$.
  For such a point,
  \[
    D(\vect w^{(0)})=0.
  \]
  Every integer point outside the box has $D(\vect w)\ge1$.

  By the lexicographic comparison property established in (\ref{eq:jcompl}), every point outside the box is strictly worse than every point inside the box.
  Consequently,
  \[
    \vect w^\star\in[\mathbf{0},\vect R].
  \]

  Among points with $D(\vect w)=0$, the comparison key first minimizes the residual $r(\vect x_*+M\vect w)$ and then the convex objective.

  Thus, setting
  \[
    \vect x_q=\vect x_*+M\vect w^\star
  \]
  and computing $\vect y_{\vect q}=\vect y_{\vect x_{\vect q}}$ by \Cref{lem:sc-local-oracle}, the pair $(\vect x_{\vect q},\vect y_{\vect q})$ minimizes \eqref{eq:sc-lex-pair} lexicographically over all admissible vectors in the residue class.

  This completes the proof.
\end{proof}

\subsection{Proof of the Main Theorem}
In order to prove \Cref{thm:sc-4block-comparison}, we enumerate the residue  of the global variables modulo $M=kN$.
In each nonempty residue class, \Cref{lem:sc-residue-optimization} computes a pair \((\vect x_{\vect q}, \vect y_{\vect q})\), which minimizes the constraint residual first and the separable convex objective second.

Comparing these pairs globally then yields either an optimal feasible solution or a certificate that no feasible solution exists.

\begin{proof}[Proof of Theorem~\ref{thm:sc-4block-comparison}]

  Let $N$ be the dilation from \Cref{prop:idp} for the generalized $n$-fold matrix $Q=[E\ I]$, and set $M=kN$.

  For every residue
  \[
    \vect q\in\{0,\ldots,M-1\}^k
  \]
  apply \Cref{lem:sc-residue-optimization}.

  If the corresponding global box is empty, discard the residue.

  Otherwise, obtain $(\vect x_{\vect q},\vect y_{\vect q})$ minimizing
  \[
    \left(
    \|\vect b-Z\vect x-E\vect y\|_1,\,
    \Phi(\vect x,\vect y)
    \right)
  \]
  lexicographically among all pairs in that residue class.

  Compare the resulting pairs across all nonempty residue classes.
  Let
  $(\vect x^\star,\vect y^\star)$ denote a lexicographically minimal pair.

  We claim that the original 4-block instance is feasible if and only if
  \[
    \vect b-Z\vect x^\star-E\vect y^\star=\mathbf{0}.
  \]

  If the original instance is feasible, there exists a pair $(\bar{\vect x},\bar{\vect y})$ in the product box with
  \[
    \vect b-Z\bar{\vect x}-E\bar{\vect y}=\mathbf{0}.
  \]
  The vector $\bar{\vect x}$ belongs to one of the enumerated residue
  classes.
  Lexicographic minimality within that residue, followed by global lexicographic minimality, implies that the minimum residual is zero.
  Hence $(\vect x^\star,\vect y^\star)$ is feasible.

  Conversely, if $\vect b-Z\vect x^\star-E\vect y^\star=\mathbf{0}$, the returned pair is feasible.

  If the returned residual is positive, no feasible pair can exist, since any feasible pair would have residual zero and would be lexicographically smaller.

  Now suppose the returned residual is zero.
  Let $(\bar{\vect x},\bar{\vect y})$ be any feasible solution.
  Both pairs have residual zero, and global lexicographic optimality gives
  \[
    \Phi(\vect x^\star,\vect y^\star)
    \le\Phi(\bar{\vect x},\bar{\vect y}).
  \]
  Thus $(\vect x^\star,\vect y^\star)$ is an optimal solution of the separable convex 4-block integer program.

  It remains to analyze the running time.

  The number of residue classes is
  \[
    M^k=(kN)^k.
  \]
  By \Cref{lem:sc-residue-optimization}, each residue is processed in time
  \[
    f_1(k,\bar\Delta)\cdot (|I|+\log N)^{O(1)}.
  \]

  The lexicographic comparison between the solutions of two residues requires only comparison of their integral residuals and at most two calls to the original objective oracle.
  Since the objective oracle gets only called for values $\vect z,\vect z'$ with $\enc(\vect z,\vect z') \leq f_2(k,\bar\Delta)\cdot |I|^{O(1)}$ for a computable function $f_2$, we get as overall running time
  \[
    (kN)^k f_1(k,\bar\Delta)\cdot (|I|+\log N)^{O(1)}(|I| + f_2(k,\bar\Delta)\cdot |I|^{O(1)})^{O(1)}.
  \]

  By \Cref{prop:idp}, $N$ is bounded by a function in $k$ and $\bar\Delta$.
  Absorbing $(kN)^k$, $\log N$, and the remaining parameter dependence into a computable function $f$ gives
  \[
    f(k,\bar\Delta)\cdot |I|^{O(1)}.
  \]

  The algorithm returns an optimal solution or correctly reports infeasibility, completing the proof.
\end{proof}

\end{document}